%% file: main.tex
\documentclass[letterpaper, 10 pt, conference]{ieeeconf} % Comment this line out
\IEEEoverridecommandlockouts % This command is only
\usepackage{adjustbox}
\usepackage{algorithm}
\usepackage{algpseudocode}
\usepackage{amsfonts}
\usepackage{amsmath}
\usepackage{amsthm}
\usepackage{amssymb}
\usepackage{bm}
\usepackage{booktabs}
\usepackage{cite}
\usepackage{color, colortbl}
\usepackage{diagbox}
\usepackage[hidelinks]{hyperref}
\usepackage{cleveref}
\usepackage{makecell}
\usepackage{mathrsfs}
\usepackage{mathtools}
\usepackage{multirow, multicol}
\usepackage{slashbox}
\usepackage{threeparttable}
\usepackage[normalem]{ulem}
\usepackage{xcolor}

\usepackage{graphicx}
\usepackage{subcaption} % For subfigure environment
\algrenewcommand\textproc{}% Used to be \textsc
\newcolumntype{M}[1]{>{\centering\arraybackslash}m{#1}}
\newcolumntype{C}[1]{>{\centering\arraybackslash}m{#1}}

\theoremstyle{plain}

\newtheorem{proposition}{\textbf{Proposition}}
\newtheorem{theorem}{\textbf{Theorem}}

\newtheorem{assumption}{Assumption}

\theoremstyle{definition}
\newtheorem{definition}{\textbf{Definition}}
\newtheorem{problem}{\textbf{Problem}}
\newtheorem{remark}{Remark}
\newtheorem{example}{\textbf{Example}}
\useunder{\uline}{\ul}{}

\title{\LARGE \bf Backstepping Reach-avoid Controller Synthesis for Multi-input Multi-output Systems with Mixed Relative Degrees}

\author{Jianqiang Ding, Dingran (Tintin) Yuan and Shankar A. Deka,
\IEEEmembership{Member, IEEE}
\thanks{ Department of Electrical Engineering, Aalto University, Finland. Email:
{\tt\small \{jianqiang.ding, tintin.yuan, shankar.deka\}}@aalto.fi.}
}

\begin{document}
    \maketitle
    \thispagestyle{empty}
    \pagestyle{empty}

    \input{abstract}
    \input{introduction_with_related_works}
    \input{preliminaries}
    \input{method}
    \input{results}
    \input{conclusion}

    \bibliographystyle{ieeetr}
    \bibliography{refs}

\end{document}

%% file: abstract.tex
\begin{abstract}
    % With the increasing complexity of cyber-physical systems in safety critical applications, designing controllers with provable formal guarantees has become an urgent requirement across various fields.
    Designing controllers with provable formal guarantees has become an urgent requirement for cyber-physical systems in safety-critical scenarios.
    Beyond addressing scalability in high-dimensional implementations, 
    controller synthesis methodologies separating safety and reachability objectives may risk optimization infeasibility due to conflicting constraints,
    thereby significantly undermining their applicability in practical applications.
    % and scalability for high-dimensional systems is also a critical metric in evaluating these methodologies. 
    % To address the risk of failure in the optimization of prevailing frameworks that consider safety and reachability certificates separately, as well as the poor scalability of provable formal certificate synthesis approaches, such as those relying on sum-of-squares optimization techniques in high-dimensional dynamical systems.
    In this paper,
    by leveraging feedback linearization and backstepping techniques, we present a novel framework for constructing provable reach-avoid formal certificates tailored to multi-input multi-output systems. Based on this, we developed a systematic synthesis approach for controllers with reach-avoid guarantees, which ensures that the outputs of the system eventually enter the predefined target set while staying within the required safe set.
    Finally, we demonstrate the effectiveness of our method through simulations.
\end{abstract}

%% file: introduction_with_related_works.tex
\section{Introduction}

% Motivation: Backstepping+ Feedback linearization + Reach-avoid Controller Synthesis

% To synthesize reach-avoid controllers for high-dimensional systems, especially when the reach-avoid constraints are defined for the output of the system, to narrow the gap between theory and applications.

% The escalating complexity of CPS in real-world applications, coupled with their need to accomplish more sophisticated tasks, is imposing stricter formal requirements on controller design methodologies for modern control systems, particularly in safety-critical scenarios.
The increasing complexity of cyber-physical systems (CPS) in real-world applications, along with their demand for more sophisticated tasks, imposes stricter formal requirements on controller design, especially in safety-critical scenarios.
In such contexts, controllers must not only avoid unsafe states but also steer the system to reach the predefined objectives.
This dual requirement, which necessitates formal guarantees that systems initialized from safe configurations will persistently maintain safety until achieving designated goals, has positioned the synthesis of controllers guaranteeing reach-avoid property as a pivotal research focus within the control community.
% This dual requirement—ensuring systems starting in safe states remain safe while reaching designated goals—has made reach-avoid controller synthesis a key research focus in control theory.

% A prominent methodology to achieving formal guarantees involves designing energy-like value functions that quantify the deviation of a system's behaviors from desired requirements, thereby enabling real-time evaluation of safety or stability violations.
% Due to their straightforward integration as additional constraints into specific control problem frameworks such as model predictive control (MPC), such approaches have been referred to safety filters and have emerged as a key paradigm for synthesizing controllers with formal guarantees in recent years.

A prominent approach to achieving formal guarantees is through energy-like scalar functions that 
make statements about the system's long-term behavior
% quantify the system's deviation from desired behavior 
\cite{choi2021robust, ames2016control, ames2019control, nguyen2021robust,  deka2018robust}, enabling real-time safety or stability evaluation. These functions can integrate easily as constraints in control frameworks like model predictive control (MPC), leading to the rise of safety filters as a key paradigm for controller synthesis with formal guarantees \cite{wabersich2023data, zeng2021safety, yuan2024safe}.
% Among these developments, control Lyapunov functions (CLFs) and control barrier functions (CBFs) \cite{xu2018constrained} stand out as the most noteworthy methods for their mathematical rigor and compatibility with real-time optimization, and make them pivotal for safety-critical control in robotics, power systems, and beyond.
% This has further led to the unification of CLFs and CBFs through quadratic program (QP) frameworks \cite{ames2019control}\cite{ames2016control}, enabling the achievement of control objectives while adhering to conditions that ensure the system stays within the safe set.
% However, a significant challenge for applying these approaches lies in identifying suitable value functions to serve as corresponding formal certificates, particularly barrier functions for ensuring safety.
Among these developments, control Lyapunov functions (CLFs) and control barrier functions (CBFs) \cite{ames2016control, ames2019control, nguyen2021robust} are notable for their mathematical rigor and real-time optimization compatibility, making them essential for safety-critical control in robotics, power systems, and beyond. Their unification through quadratic program (QP) frameworks \cite{ames2016control, ames2019control} enables control objectives while ensuring stability and safety specifications. However, a key challenge remains in identifying suitable scalar functions for formal safety certificates.

% , often presents a significant challenge.
% While computational techniques such as SOS optimization and Hamilton-Jacobi reachability analysis  \cite{zhang2023efficient, choi2021robust, bansal2017hamilton} have demonstrated their applicability to calculate the value function and invariant set in lower-dimensional systems, their scalability limitations motivate ongoing research for complex high-dimensional systems.
% For lower-dimensional systems, computational techniques including sum-of-squares (SOS) optimization and Hamilton-Jacobi reachability analysis have emerged as effective solutions, enabling systematic construction of these certificates while maintaining tractability.
% Recent advances such as backstepping-based formal certificate construction approaches \cite{vaidyanathan2020backstepping} have significantly enhanced the scalability of these methodologies when establishing formal guarantees for high-dimensional dynamical systems. Examples include using backstepping with CBF and CLF to ensures set invariance for high-order systems \cite{taylor2022safe, kim2023safe, koga2023safe}, with applications like multirotor flight control. The main limitations of these methods are the structural requirements of the dynamics and the controllability assumptions, which prevent direct applicability to systems lacking of cascaded architecture.

Computational techniques such as sum-of-squares (SOS) optimization and Hamilton-Jacobi reachability analysis \cite{zhang2023efficient, bansal2017hamilton, chen2025distributionally} have proven effective for computing value functions and invariant sets in low-dimensional systems. However, their scalability limitations have spurred ongoing research for complex, high-dimensional systems. Recent advances, such as backstepping-based formal certificate construction \cite{vaidyanathan2020backstepping}, have significantly improved scalability in establishing formal guarantees for high-dimensional dynamics. For instance, combining backstepping with control barrier functions (CBFs) and control Lyapunov functions (CLFs) ensures set invariance for high-order systems \cite{taylor2022safe, kim2023safe, koga2023safe}, with applications in multirotor flight control. Despite these advancements, these methods are limited by their reliance on specific dynamic structures and controllability assumptions, restricting their applicability to systems without a cascaded architecture. % \cite{deka2018robust} proposed a method that uses avoidance value functions and state feedback in the joint space of robot manipulators, but has a limitation of extending to more general geometry.

% High order CBFs (HOCBFs), which also establish forward invariance guarantees for safe ..., recursively differentiate safety constraints until control inputs explicitly appear in the derivatives, and then apply CBF-inspired constraints on the highest-order derivatives, thereby eliminates the structural requirements by backstepping methods.
% While other methods for constructing safe invariant set for high order systems like exponential CBFs and high order CBFs (HOCBFs) \cite{xiao2021high, nguyen2016exponential} similarly construct a safe set within the set defined by the prescribed constraints to ensure forward invariance, they recursively differentiate safety constraints until control inputs explicitly encoded in the highest-order derivatives, and then apply CBF-inspired conditions on these derivatives, thereby eliminates the structural requirements of backstepping methods under the assumption of uniform relative degree.
% This conceptual alignment with feedback linearization has inspired synthesizing control barrier functions for partially feedback linearizable systems \cite{cohen2024constructive}, enabling their deployment in underactuated robotic systems.
Alternative approaches, such as exponential CBFs and high-order CBFs \cite{nguyen2016exponential, xiao2019control, xiao2021high}, address some of these limitations by recursively differentiating safety constraints until control inputs explicitly appear in the highest-order derivatives. These methods then enforce CBF-inspired conditions on these derivatives, eliminating the need for cascaded structures under the assumption of a uniform relative degree. This process is conceptually aligned with feedback linearization for high-relative-degree systems and has inspired the synthesis of CBFs for partially input-output linearizable systems \cite{cohen2024constructive, xu2018constrained, abel2023prescribed}, with applications in safety control for underactuated robotic systems.
% While these methods offer greater flexibility, challenges remain in balancing computational complexity and ensuring robustness in real-world applications. 
Similarly, recent learning-based techniques \cite{dawson2022learning, hsu2021safety} leverage memory of object locations in environments but often rely on slack variables to relax goal-reaching objectives in new or uncertain scenarios. Both approaches highlight the need for more integrated frameworks that jointly address safety and stability without compromising computational feasibility. However, treating safety and stability formal guarantees as separate certificates risks computational infeasibility during optimization due to inherent conflicts between these constraints.

% Recent approaches have also turned to learning-based techniques to synthesize control strategies \cite{dawson2022learning, hsu2021safety}. While these methods essentially remember where the objects are in the environments, they still rely on slack variables to relax goal-reaching objectives when encountering new or uncertain environments, highlighting the need for more integrated approaches that jointly address safety and stability without compromising computational feasibility.
% However, these frameworks that treat safety and stability formal guarantees as separate certificates may risk computational infeasibility during optimization due to inherent conflicts between these constraints. 

In contrast, recent developments in optimization-based reach-avoid controller synthesis propose a unified methodology with rigorous safety guarantees\cite{xue2024reach} to address this challenge. This approach originates from the characterization of the reach-avoid set \cite{xue2023reach} through the zero-level set of the synthesized scalar functions, which verifies system compliance with reachability and safety requirements. Beyond simplifying design procedures by eliminating redundant value functions, this method can ensure persistent provable reach-avoid guarantees throughout the evolving of the dynamics, thereby preventing the infeasibility of the problem due to constraint conflicts during the optimization.

In practice, multi-input multi-output (MIMO) systems are quite common in daily life and safety constraints are often imposed on observable outputs rather than system states. Motivated by this output-to-input relationship and aforementioned challenges, we in this paper present a systematic reach-avoid controller synthesis framework for MIMO systems. Our main contributions can be summarized as follows.
\begin{itemize}
    \item We propose a constructive framework for generating formal reach-avoid certificates in MIMO systems,
    which characterize a subset of the given safe set
    where all states within it can achieve the specified reach-avoid objectives with synthesized controllers.
    % with uniform relative degrees, and derive a systematic synthesis methodology for controllers that guarantee reach-avoid specifications.
    \item 
    % Furthermore, we generalize this framework to address MIMO systems with mixed relative degrees and element-wise output constraints. 
    Building upon this foundation, we establish a reach-avoid controller synthesis approach for polynomial dynamical systems, and demonstrate the proposed 
    % reach-avoid controller synthesizing 
    pipeline in numerical simulations.
    % and rigorously characterize its feasibility conditions.
\end{itemize}
% \begin{itemize}
%     \item We propose a framework for constructing provable reach-avoid formal certificates for MIMO systems with a uniform relative degree. Based on this formal certificate, we then present a systematic controller synthesizing method with reach-avoid formal guarantees.
%     \item Furthermore, we extend the aforementioned formal certificates construction framework to MIMO systems with mixed relative degrees and element-wise constraints on their outputs. We then establish the reach-avoid controller synthesizing approach and discuss the condition for its feasibility.
% \end{itemize}

\textit{Notations:} We use $\mathbb{R}$ to denote the set of real numbers, and $\mathbb{R}
^{n}$ denotes the $n-$dimensional Euclidean space. A $\mathcal{C}^{\infty}(\mathcal{X})$ function
refers to the set of smooth functions defined 
% a function that has derivatives of all orders 
on the domain $\mathcal{X}$. The
closure of a set $\mathcal{X}$ is denoted by $\overline{\mathcal{X}}$.
% The space of all $k-$times continuously differentiable functions on domain $X$ is denoted by $\mathcal{C}^k(X)$, and
We denote column vectors using lowercase boldface letters and matrices with upper boldface letters, such as $\bm{v}$ and $\bm{M}$.
% $\bm{M}^{\dagger}$ denotes the Moore-Penrose inverse of a matrix $\bm{M}$.
% Let $\prod_{i=1}^n S_i$ be the Cartesian product of $S_1, \cdots, S_n$.
% For each $j \in \{1,\cdots, n\}$, the $j$th projection on $S=\prod_{i=1}^n S_i$ is the mapping $\text{pr}_j : S \rightarrow S_j$ defined by,
% \begin{align*}
%     \text{pr}_j s = s_j, \quad \forall s \in S
% \end{align*}.
% The vector $\bm{e}_{i}$ denotes the unit column vector with the $i$th element being $1$.
% The Cartesian product of two sets $A$ and $B$, denoted $A \times B$.
% Let $\Omega_p$ denotes the euclidean space of dimension $p$.
The ring of all multivariate polynomials in a variable $\bm{x}$ is denoted by $\mathbb{R}[\bm{x}]$. $\sum [\bm{x}]$ is used to represent the sum-of-squares polynomials over variable $\bm{x}$, i.e.,
\begin{align*}
    \sum [\bm{x}] = \{ p \in \mathbb{R}[\bm{x}] \,|\, p = \sum_{i=1}^{k}q_{i}^{2}, q_{i}\in \mathbb{R}[\bm{x}], i=1,\cdots, k \}.
\end{align*}

% \section{Related work}

% Other methods to design safety-critical stabilising controller for high-order nonlinear systems:

% exponential approach, for example the exponential barrier function \cite{nguyen2016exponential, nguyen2021robust},

% high-order control barrier function, by constructing sublevel set for each derivative layer \cite{xiao2019control,xiao2021high}.

% \textit{Feedback linearisation}

% In this paper, we use the structure derived from feedback linearisation \cite{khalil2002nonlinear} to aid the backstepping-based synthesis of a reach-avoid controller.

% (Xiangru X)\cite{xu2018constrained} \cite{abel2023prescribed} treating safety constraints directly as outputs for strict-feedback nonlinear systems.

% Cohen et al. used the cascaded structure of feedback linearisation to construct valid barrier functions for nonlinear systems \cite{cohen2024constructive}.

% \textit{Reach-avoid}

% Learning safe, generalisable perception-based reach-avoid \cite{dawson2022learning}.

% Controller synthesis for linear systems with reach-avoid specification \cite{fan2018controller} \cite{fan2021controller}.

% CBF reach-avoid, combining CBF with HJR, computationally expensive \cite{choi2021robust}

% Reach-avoid problem with time varying dynamics, targets and constraints \cite{fisac2015reach}.

% Reach-avoid via reinforcement learning \cite{hsu2021safety}.

% \cite{deka2018robust} proposed a method that uses state feedback in the joint space of robot manipulators, but has a limitation of extending to more general geometry.

%% file: preliminaries.tex
\section{Preliminaries and Background}

% \textcolor{red}{TODO}

% \color{blue}
% \begin{itemize}
%     \item System Definition
%     \item  Problem Statement
%     \item Feedback linearization
%     \item Control Barrier Function Backstepping
%     \item Control Guidance-barrier Functions \cite{xue2024reach}
% \end{itemize}

% \color{black}

Consider a nonlinear system
% \begin{align} \label{eq: system}
%     \dot{\bm{x}} &= f(\bm{x}) + g(\bm{x})\bm{u} \\
%     \bm{y} &= h(\bm{x})
% \end{align}
% \begin{equation}
%     \label{eq: system}% \tag{$\star$}
%     \begin{aligned}
%         \dot{\bm{x}} & = f(\bm{x}) + \sum_{j=1}^{m}g_{j}(\bm{x})\bm{u}_{j} \\
%         \bm{y}       & = h(\bm{x})
%     \end{aligned}
% \end{equation}
\begin{align}
    \dot{\bm{x}} & = \bm{f}(\bm{x}) + \sum_{j=1}^{m}\bm{g}_{j}(\bm{x})u_{j}\label{eq: dynamic system} 
    % \\
    % \bm{y}       & = \bm{h}(\bm{x}) \label{eq: system output}
\end{align}
with $\bm{f}: \mathcal{X}\rightarrow \mathbb{R}^{n}$ and
$\bm{g}_{j}: \mathcal{X}\rightarrow \mathbb{R}^{n},~j\in\{1,\cdots, m\}$ are
locally Lipschitz, $\bm{x}\in \mathcal{X}\subseteq \mathbb{R}^{n}$ and admissible
control input $u_{j}\in \mathbb{R}, j\in\{1, \cdots, m\}$ such that the closed-loop system \eqref{eq: dynamic system} is locally Lipschitz.
% Then for any initial state $\bm{x}_0 \in \mathbb{R}^n$, there exists a maximal time interval $I(\bm{x}_0) = [0,t_{\max}(\bm{x}_0))$ such that a unique continuously differentiable solution $\varphi_{\bm{x}_0}(0)=0$ and $t_{\max}$ is the explosion time with $\lim_{t \rightarrow t_{\max}} \varphi_{\bm{x}_0}(t) = \infty$.
% \textcolor{red}{(here $u$ may no need to be constrained)}
% $\bm{y}\in \mathbb{R}^{m}$ is the output of the system with a smooth function $\bm{h}: \mathcal{X}\rightarrow \mathbb{R}^{m}$ of class $\mathcal{C}^{\infty}(\mathcal{X})$, and all Lie derivatives of $h$ with respect to $\bm{f}(\bm{x})$, denoted as $\mathcal{L}_{\bm{f}}^k h(\bm{x}), \forall k \in \mathbb{Z}_{\geq 1}$, are non-constant functions. 

Let $\bm{y} \in \mathbb{R}^m$ be the outputs of the system \eqref{eq: dynamic system} defined by a smooth function $h: \mathcal{X} \rightarrow \mathbb{R}^m$ of class $\mathcal{C}^{\infty} (\mathcal{X})$. 
% Given a bounded open safe set $\mathcal{C}$ and a target set $\mathcal{X}^r$,
We formulate the output reach-avoid controller synthesis problem as follows.
% \TTY{changed the notation of $\bm{u}$ and $g$.}
\begin{problem}
    (Reach-avoid Controller Synthesis) 
    Given a safe set $\mathcal{C}$ and a target set $\mathcal{X}^r$, both defined through the constraints on output $\bm{y}$, our objective is to construct a set $\mathcal{S} \subseteq \mathcal{C}$ and a system state feedback controller $\bm{k}(\bm{x}): \mathcal{X} \rightarrow \mathbb{R}^m$ such that, starting from any initial state within $\mathcal{S}$, the trajectory of the system \eqref{eq: dynamic system} with controller $\bm{k}(\bm{x})$ remains within $\mathcal{C}$ and eventually enters $\mathcal{X}^r$.
    % Given the system \eqref{eq: system without output}
    % with output \eqref{eq: system output}, the safe output set
    % $\mathcal{C}\subseteq \mathbb{R}^{m}$ is defined by a smooth function
    % $\psi: \mathbb{R}^{m}\rightarrow \mathbb{R}$ as
    % \begin{align}
    %     \mathcal{C}:= \{y \in \mathbb{R}^{m}\ | \ \psi(\bm{y}) > 0\} \label{eq: safe output set}
    % \end{align}
    % and the target output set $\mathcal{Y}^{r}\subseteq \mathbb{R}^{m}$ is
    % defined by a smooth function $\phi: \mathbb{R}^{m}\rightarrow \mathbb{R}$ as
    % \begin{align}
    %     \mathcal{Y}^{r}:= \{y \in \mathbb{R}^{m}\ | \ \phi(\bm{y}) < 0\} \label{eq: target output set}
    % \end{align}
    % In this paper, we consider the reach-avoid controller synthesis problem, which aims to find a feedback controller $\bm{u}= \bm{k}(\bm{y})$ that can drive the system
    % output \eqref{eq: system output} to enter the target output set $\mathcal{Y}_{r}$
    % without going into unsafe output states out of $\mathcal{C}$.
\end{problem}

\begin{assumption}
    \label{assumption on sets} 
    Furthermore, we make the following assumptions:
    % We make the following assumptions on the safe output set $\mathcal{C}$ and target output set $\mathcal{Y}^{r}$.
    \begin{itemize}
        \item  $\mathcal{C} = \{\bm{x} \in \mathcal{X} | \ \psi(\bm{y}(\bm{x})) > 0 \}$
        % \item $\overline{\mathcal{C}} = \{x \in \mathbb{R}^{n}\ | \ \psi(\bm{y}(\bm{x})) \geq 0 \}$
        \item $\mathcal{X}^r = \{ \bm{x} \in \mathcal{X} | \phi (\bm{y}(\bm{x})) < 0 \}$
        \item $\mathcal{C}$ has no isolated point
        \item $\mathcal{C}\cap \mathcal{X}^{r}$ is not empty and has no isolated point,
    \end{itemize}
    where $\psi: \mathbb{R}^m \rightarrow \mathbb{R}$ and $\phi: \mathbb{R}^m \rightarrow \mathbb{R}$ are both continuously differentiable functions.
\end{assumption}

% Additionally, as we are considering constraints over a space constructed through
% Cartesian product operations, we proposed the following proposition regarding the
% existence of isolated points in the Cartesian product space.

% \begin{proposition}
%     (Product Topology \cite{munkres2000topology}) \label{prop: isolated points in
%     product space} Let $\mathcal{M}$ be a metric space, then the point
%     $q = (q_{1}, \cdots, q_{n})$ is an isolated point of
%     $\mathcal{S}\subseteq \mathcal{M}= \mathcal{M}_{1}\times \cdots \times \mathcal{M}
%     _{n}$, if and only if, each $q_{i}$ is an isolated point of a set
%     $\mathcal{S}_{i}$ in $\mathcal{M}_{i}$.
% \end{proposition}
% (\textcolor{red}{Need references for this proposition})

\subsection{Reach-avoid Controller Synthesis}
\label{subsec: reach-avoid controller synthesis}

% Consider an affine control system
% \begin{align}
%     \dot{y} = F(y) + G(y) v
% \end{align}
% with $F$ and $G$ are locally Lipschitz, $x \in \mathbb{R}^n$ and $v \in \mathcal{V} \subseteq \mathbb{R}^m$ is the set of admissible control inputs.

\begin{definition}
    (Exponential Control Guidance-barrier Functions (ECGBFs)) Given the safe set $\mathcal{C}$ and target set $\mathcal{X}^{r}$ satisfying assumption \ref{assumption on sets},
    $\psi(\bm{y}(\bm{x})): \mathbb{R}^{n}\rightarrow \mathbb{R}$ is an ECGBF if there exists $\lambda>0$ such that
    \begin{align} \label{eq: ECGBF constraint}
        % \sup_{u(y) \in \mathcal{U}}\frac{\partial \psi(y)}{\partial y}\dot{y}(u) \geq \lambda \psi(y), \ \ \forall y \in \overline{\mathcal{C}\setminus \mathcal{X}^r}
        \sup_{\bm{u}}\mathcal{L}_{\psi,\bm{u}}\geq \lambda \psi(\bm{y}(\bm{x})), \forall \bm{x} \in \overline{\mathcal{C}\setminus \mathcal{X}^r}, % sup_{u(\bm{x})} \nabla_{\bm{x}} h(\bm{x}) f(\bm{x}) + \nabla_{\bm{x}} g(\bm{x}) u(\bm{x}) \geq \lambda \psi_x(x), \forall x \in \overline{\mathcal{C}_x\setminus \mathcal{X}^r}
    \end{align}
    where
    $\mathcal{L}_{\psi,\bm{u}}= \frac{\partial \psi(\bm{y}(\bm{x}))}{\partial \bm{x}} \cdot \bm{f}(\bm{x}) +
    \sum_{j=1}^k \frac{\partial \psi(\bm{y}(\bm{x}))}{\partial \bm{x}}  \bm{g}_j(\bm{x}) \bm{u}_j(\bm{x})$.
    % where $L_{V,u}(y) = \nabla_y V(y) \cdot \dot{y}$
\end{definition}

\begin{theorem}\cite{xue2024reach} \label{theorem: reach-avoid controller synthesis}
    Given the safe output set $\mathcal{C}$ and target output set $\mathcal{X}^{r}$
    satisfying assumption \ref{assumption on sets}, if the function
    $\psi(\bm{y}(\cdot)): \mathbb{R}^{n}\rightarrow \mathbb{R}$ is an ECGBF, then any Lipschitz continuous controller $\bm{u}(\bm{x}) \in \mathcal{K}_{e}(\bm{x})$ is a reach-avoid controller with respect to the safe set $\mathcal{C}$ and target set $\mathcal{X}^{r}$ where
    \begin{align}
        \mathcal{K}_{e}(\bm{x}) = \{\bm{u}(\bm{x})\ \in \mathbb{R}^m | \ \text{constraints \eqref{eq: reach-avoid controller} holds}\}
    \end{align}
    with
    \begin{equation}
        \label{eq: reach-avoid controller}\left\{
        \begin{split}
            % \frac{\partial \psi(y)}{\partial y}\dot{y}&\geq \lambda \psi(y), \ \
            % \forall y \in \overline{\mathcal{C}\setminus \mathcal{X}^r} \\ 
            \mathcal{L}_{\psi,\bm{u}(\bm{x})} &\geq \lambda \psi(\bm{y}(\bm{x})), \forall \bm{x}\in \overline{\mathcal{C}\setminus \mathcal{X}^r} \\
            \lambda &> 0.
        \end{split}
        \right.
    \end{equation}
\end{theorem}

For polynomial systems with semialgebraic safe set and target set, the reach-avoid controller synthesis reduces to solving the following convex optimization problem,
% When the data involved are polynomials, i.e., the safe set and target set are semialgebraic, and the system has polynomial dynamics, the feedback reach-avoid controller can be determined by solving the following optimization problem,
% \begin{algorithm} 
    \begin{align*}
    &\bm{k}(\bm{x}) = \arg\min_{\bm{u}(\bm{x})} \quad \delta \\
    \text{s.t.} \quad 
    % \left\{
        % \begin{split}
    %     &\mathcal{L}_{\psi,\bm{u}}
    % - \lambda \psi(\bm{y}(\bm{x})) - s_0(x) \psi(\bm{y}(\bm{x})) - s_1(x) \phi(\bm{y}(\bm{x})) \in \sum [\bm{x}] \\
    &\mathcal{L}_{\psi,\bm{u}} - \lambda \psi(\bm{y}(\bm{x})) \geq -\delta, \forall \bm{x} \in \overline{\mathcal{C} \setminus \mathcal{X}^r} \\
    & \delta \geq 0 \\
    &\lambda \geq \epsilon
    \end{align*}
    % \begin{align}
    %     % \nabla_{\bm{x}}\psi_x(\bm{x}) \cdot \bm{f}(\bm{x}) + \sum_{j=1}^k \nabla_{\bm{x}} \psi_x\bm{g}_j(\bm{x}) \bm{u}_j(\bm{x})
    % \mathcal{L}_{\psi,\bm{u}}
    % - \lambda \psi(\bm{y}(\bm{x})) - s_0(x) \psi(\bm{y}(\bm{x})) - s_1(x) \phi(\bm{y}(\bm{x})) \in \sum [\bm{x}]
    % \end{align}
    where 
    % $s_0(x),s_1(x) \in \sum[\bm{x}]$, 
    $\epsilon > 0$ is a user-defined factor. 
% \end{algorithm}
The problem can be formulated and solved via SOS optimization technique as follows:
\begin{algorithm}  
    \begin{align} \label{algo: sos reach-avoid controller}
    &\bm{k}(\bm{x}) = \arg\min_{\bm{u}(\bm{x})} \quad \delta \notag \\
    \text{s.t.} \quad 
        % \nabla_{\bm{x}}\psi_x(\bm{x}) \cdot \bm{f}(\bm{x}) + \sum_{j=1}^k \nabla_{\bm{x}} \psi_x\bm{g}_j(\bm{x}) \bm{u}_j(\bm{x})
    &\mathcal{L}_{\psi,\bm{u}}
    - \lambda \psi + \delta - s_0(\bm{x}) \psi - s_1(\bm{x}) \phi \in \sum [\bm{x}]  \notag \\
    &\delta> 0  \\
    &\lambda > \epsilon. \notag
    \end{align} 
    where $s_0(\bm{x}),s_1(\bm{x}) \in \sum[\bm{x}]$.
\end{algorithm}

If $\delta=0$, then $\bm{k}(\bm{x})$ is a reach-avoid controller with respect to the safe set $\mathcal{C}$ and target set $\mathcal{X}^r$.

\subsection{Feedback Linearization}

Consider a MIMO control-affine nonlinear system of the form \eqref{eq: dynamic system} with output $\bm{y}= h(\bm{x})$. Differentiating the $i^\text{th}$ output $y_{i}$ with respect to time yields
\begin{align}
    \dot{y}_{i}                                      % &=
    % \begin{bmatrix}
    %     \frac{\partial h_i}{\partial x_1} & \cdots &  \frac{\partial h_i}{\partial x_n}
    % \end{bmatrix} [f(x) + g(x) \bm{u}] \\
    % &=
    % \begin{bmatrix}
    %     \frac{\partial h_i}{\partial x_1} & \cdots &  \frac{\partial h_i}{\partial x_n}
    % \end{bmatrix} f(x) +
    % \begin{bmatrix}
    %     \frac{\partial h_i}{\partial x_1} & \cdots &  \frac{\partial h_i}{\partial x_n}
    % \end{bmatrix} g(x) \bm{u} \\
    % &= L_f h_i + L_g h_i \bm{u}
    % &
    = \mathcal{L}_{f}h_{i}(\bm{x})+ \sum_{j=1}^{m}(\mathcal{L}_{g_j}h_{i}(\bm{x})) u_{j}.
\end{align}
Observe that if $\mathcal{L}_{g_j}h_{i}(\bm{x})= 0$ for all $j = 1, \cdots, m$, then the input does
not appear in $\dot{y}_{i}$. Assume that $y_{i}$ has to be differentiated with
respect to time $r_{i}$ times before at least one component of the control input
vector $u$ explicitly appears in a time derivative of $y_{i}$, then the $r_{i}^\text{th}$
derivative of $y_{i}$ is given by
\begin{align}
    y_{i}^{(r_i)}                                                         % &=
    % \begin{bmatrix}
    %     \frac{\partial L_f^{r_i -1} h_i}{\partial x_1} & \cdots &  \frac{\partial L_f^{r_i -1} h_i}{\partial x_n}
    % \end{bmatrix} [f(x) + g(x) u] \\
    % &
    % = L_f^{r_i} h_i + L_g (L_f^{r_i -1} h_i) u
    = \mathcal{L}_{f}^{r_i}h_{i}(\bm{x})+ \sum_{j=1}^{m}\mathcal{L}_{g_j}(\mathcal{L}_{f}^{r_i -1 }h_{i}(\bm{x})) u_{j}.
\end{align}
The integer $r_{i}$ is defined as the smallest integer such that
\begin{align}
    \mathcal{L}_{g_j}L_{f}^{k}h_{i}(\bm{x})     & = 0, \quad 1 \leq j \leq m, 0 \leq k \leq r_{i}-2 \label{eq: relative degree 0}     \\
    \mathcal{L}_{g_j}\mathcal{L}_{f}^{r_i-1}h_{i}(\bm{x}) & \neq 0, \quad \text{for at least one }1 \leq j \leq m. \label{eq: relative degree 1}
\end{align}
For specific $h_{i}(\bm{x})$, we can define
\begin{align} \label{eq: eta_i}
    \bm{\eta}^{i}= 
    \begin{bmatrix}
    \eta_{1}^{i}\\ 
    \vdots \\ 
    \eta_{r_i}^{i}
    \end{bmatrix} = 
    \begin{bmatrix}
    y_{i}(\bm{x}) \\ 
    \vdots \\ 
    \mathcal{L}_{f}^{r_i-1}y_{i}(\bm{x})
    \end{bmatrix} = 
    \begin{bmatrix}
    h_{i}(\bm{x}) \\ 
    \vdots \\ 
    \mathcal{L}_{f}^{r_i-1}h_{i}(\bm{x})
    \end{bmatrix} \in \mathbb{R}^{r_i}.
\end{align}
Then we have
\begin{align}
    \dot{\bm{\eta}}^i                                                         % &=
    % \begin{bmatrix}
    %     \dot{\eta_1} \\
    %     \vdots \\
    %     \dot{\eta_{r_i}}
    % \end{bmatrix}
    = \begin{bmatrix}
    0 & 1 & 0 & \cdots & 0 \\ 
    0 & 0 & 1 & \cdots & 0 \\ 
    \vdots & \vdots & \vdots & \ddots & \vdots \\ 
    0 & 0 & 0 & \cdots & 1 \\ 
    0 & 0 & 0 & \cdots & 0
    \end{bmatrix} % _{r_i \times r_i}
    \begin{bmatrix}
    h_{i}(\bm{x}) \\
    \vdots \\
    \mathcal{L}_{f}^{r_i-1}h_{i}(\bm{x})
    \end{bmatrix} + 
    \begin{bmatrix}
    0 \\ 
    \vdots \\
    0 \\ 
    1
    \end{bmatrix} v_{i},
\end{align}
where
$v_{i}= \mathcal{L}_{f}^{r_i}h_{i}(\bm{x})+ \sum_{j=1}^{m}\mathcal{L}_{g_j}(L_{f}^{r_i -1 }h_{i}(\bm{x})) u_{j}$. 
For single-input single-output (SISO) systems with $m=1$,
\eqref{eq: relative degree 0} and \eqref{eq: relative degree 1} are the definitions of the
relative degree of $\bm{y}=h(\bm{x})$, with $h: \mathbb{R}^{n}\rightarrow \mathbb{R}$. The concept of relative degree is extended to MIMO systems
systems as follows
\begin{definition}
    (Vector relative degree). The system \eqref{eq: dynamic system} has a vector relative degree $\{r_{1}, \cdots, r_{m}\}$ at a point $\bm{x}_{0}$ if
    \begin{itemize}
        \item
            \begin{align}
                \mathcal{L}_{g_j}\mathcal{L}_{f}^{k}h_{i}(\bm{x}) = 0, \quad 0 \leq k \leq r_{i}-2,
            \end{align}
            for all $1 \leq j \leq m$, for all $1 \leq i \leq m$, and for all $\bm{x}$
            in a neighborhood of $\bm{x}_{0}$.
        \item the 
        % $m \times m$ 
        matrix $\bm{A}(\bm{x})$ has rank $m$ at a neighborhood of $\bm{x}_0$, where
        % $\text{rank}\big(A(x)\big)=m$ at $x=x_{0}$, where
            \begin{align} \label{matrix A(x)}
                \bm{A}(\bm{x}) = 
                \begin{bmatrix}
                \mathcal{L}_{g_1}L_{f}^{r_1 -1 }h_{1}(\bm{x})&\cdots&\mathcal{L}_{g_m}\mathcal{L}_{f}^{r_1 -1 }h_{1}(\bm{x}) \\
                % L_{g_1} L_f^{r_2 -1 } h_2 (x) & \cdots & L_{g_m} L_f^{r_1 -1 } h_1 (x) \\
                \vdots&\ddots&\vdots \\ 
                \mathcal{L}_{g_1}\mathcal{L}_{f}^{r_m -1 }h_{m}(\bm{x})&\cdots&\mathcal{L}_{g_m}\mathcal{L}_{f}^{r_m -1 }h_{m}(\bm{x})
                \end{bmatrix}
            \end{align}
        % has rank $p$ at $x=x_0$.
    \end{itemize}
\end{definition}

A system \eqref{eq: dynamic system} with a well-defined vector relative degree can be written as
% \begin{align}
%     \dot{
%     \begin{bmatrix}
%         y_1^1 \\
%         \vdots \\
%         y_1^{r_1 -1} \\
%         \vdots \\
%         y_p^1 \\
%         \vdots \\
%         y_p^{r_p-1} \\
%         \cmidrule(lr){1-1}
%         y_1^{r_1} \\
%         \vdots \\
%         y_p^{r_p}
%     \end{bmatrix}} =
%     \begin{bmatrix}
%         0 & 1 & 0 & \cdots & 0 | 1 & 0 & \cdots & 0 \\
%         0 & 0 & 1 & \cdots & 0 | 0 & 1 & \cdots & 0 \\
%         \vdots & \vdots & \vdots & \ddots
%     \end{bmatrix}
% \end{align}
% \begin{align}
%     \begin{bmatrix}
%         y_1^1 \\
%         \vdots \\
%         y_1^{r_1} \\
%         \cmidrule(lr){1-1}
%         \vdots \\
%         \cmidrule(lr){1-1}
%         y_1^{r_1} \\
%         \vdots \\
%         y_p^{r_p}
%     \end{bmatrix} =
% \end{align}

\begin{align}
    \begin{bmatrix}
    y_{1}^{r_1}\\ 
    \vdots \\ 
    y_{m}^{r_m}
    \end{bmatrix} = 
    \begin{bmatrix}
    \mathcal{L}_{f}^{r_1}h_{1}(\bm{x})\\ 
    \vdots \\ 
    \mathcal{L}_{f}^{r_m}h_{m}(\bm{x})
    \end{bmatrix} + 
    \bm{A}(\bm{x}) 
    \underbrace{
    \begin{bmatrix}
    u_{1}\\ 
    \vdots \\ 
    u_{m}
    \end{bmatrix}
    }_{\bm{u}}
\end{align}

% Since $A(x_0)$ is non-singular, it follows that $A(x) \in \mathbb{R}^{m \times m}$ is bounded away from nonsingularity for $x$ in a neighborhood of $x_0$, meaning that $A^{-1}(x)$ has bounded norm on this neighborhood.
% \color{red}
Then the state feedback control law

\begin{align}
    \bm{u} = -\bm{A}^{-1}(\bm{x}) \begin{bmatrix}\mathcal{L}_{f}^{r_1}h_{1}\\ \vdots \\ \mathcal{L}_{f}^{r_m}h_{m}\end{bmatrix} + \bm{A}^{-1}(\bm{x}) \bm{v}
\end{align}
% \color{black}
yields the linear closed-loop system

\begin{align}
    \begin{bmatrix}y_{1}^{r_1}\\ \vdots \\ y_{m}^{r_m}\end{bmatrix} = \underbrace{\begin{bmatrix}v_{1}\\ \vdots \\ v_{m}\end{bmatrix}}_{\bm{v}}
\end{align}

\begin{remark}
    If $r_{1}+r_{2}+ \cdots + r_{m}= n$, the system \eqref{eq: dynamic system} is input-to-state or full-state feedback linearizable \cite{isidori1985nonlinear} since the set of functions
    \begin{align}
        \Phi_{k}^{i}(\bm{x}) = \mathcal{L}_{f}^{k-1}h_{i}(\bm{x}), \quad i \leq k \leq r_{i}, 1 \leq i \leq m
    \end{align}
    completely defines a local coordinate transformation at $\bm{x}_{0}$.
\end{remark}

%% file: method.tex
\section{Reach-avoid backstepping}

In this section, we introduce a constructive methodology to generate ECGBFs that guarantee the existence of reach-avoid controllers with respect to the safe set and target set under assumption \eqref{assumption on sets}. Building on these formal certificates, we develop a systematic approach to design controllers with provable guarantees that the system output can reach the given target set while avoiding unsafe outputs.

% Given the safe output set $\mathcal{C}$ and the target output set $\mathcal{X}_{r}$
% as defined by \eqref{eq: safe output set} and \eqref{eq: target output set}, our
% goal is to construct a control guidance barrier function that guarantees the
% existence of a controller that can drive the system \eqref{eq: system} enter the
% target output set $\mathcal{X}^{r}$ while staying within the safe output set $\mathcal{C}$.

% \subsection{Reach-avoid Backstepping for SISO Systems}
% \subsection{MIMO Systems with Uniform Relative Degrees}

Motivated by the methods presented in \cite{cohen2024constructive,taylor2022safe}, we kick off the introduction of the proposed framework by synthesising a reach-avoid controller 
$\bm{k}_1(\bm{y}(\bm{x}))$ 
for the single-integrator dynamical system described below,
\begin{align}
    \dot{\bm{y}} = \bm{v}. \label{eq: single integrator}
\end{align}
% where $\bm{y},\bm{v} \in \mathbb{R}^m$. 
% with $\bm{k}_1(\bm{y}(\bm{x}))$
% = \begin{bmatrix}
%     k_1^1(\bm{y}(\bm{x})) & \cdots & k_1^m(\bm{y}(\bm{x}))
% \end{bmatrix}^\top \in \mathbb{R}^m$ 

Suppose $\psi(\bm{y})$ is an ECGBF for system \eqref{eq: single integrator} with respect to safe set $\mathcal{C}$ and target set $\mathcal{X}^r$, by Theorem \ref{theorem: reach-avoid controller synthesis}, then there exists $\lambda>0$ such that
% \begin{align}
%     \frac{\partial \psi}{\partial \bm{y}} \cdot \bm{k}_1(\bm{y}(\bm{x})) \geq \lambda \psi(\bm{y}(\bm{x})), \forall \bm{x} \in \overline{\mathcal{C} \setminus \mathcal{X}^r}
% \end{align}
\begin{align}
    % \frac{\partial \psi}{\partial \bm{y}} \cdot \bm{k}_1(\bm{y}(\bm{x})) = 
    \sum_{i=1}^m \frac{\partial \psi}{\partial y_i} \cdot k_1^i(\bm{y}(\bm{x}))
    \geq \lambda \psi(\bm{y}(\bm{x})), \forall \bm{x} \in \overline{\mathcal{C} \setminus \mathcal{X}^r}
\end{align}
where $\bm{k}_1(\bm{y}(\bm{x})) = \begin{bmatrix}
    k_1^1(\bm{y}(\bm{x})) & \cdots & k_1^m(\bm{y}(\bm{x}))
\end{bmatrix}^\top \in \mathbb{R}^m$.
This implies that the controller $\bm{k}_1(\bm{y}(\bm{x}))$ can steer the system \eqref{eq: single integrator} with formal guarantees that state trajectories stay within the safe set $\mathcal{C}$ and finally enter into the target set $\mathcal{X}^r$, thus ensuring the satisfaction of reach-avoid requirements on the output $\bm{y}$.

Now we consider a MIMO system of the form \eqref{eq: dynamic system} 
with vector relative degree $\{\gamma_1, \cdots, \gamma_m \}$ at a point $\bm{x}_0 \in \mathcal{X}$.
% , and $\gamma$ is the uniform relative degree of $\bm{y}=\bm{h}(\bm{x})$. 
Let $\bm{\eta}^i$ be the transformed variable as defined in \eqref{eq: eta_i},
% \begin{align}
%     \label{eq: mapping definition}
%     \bm{\eta}^i = 
%     \begin{bmatrix}
%         \eta_{1}^i \\ 
%         \vdots \\ 
%         \eta_{\gamma_i}^i
%     \end{bmatrix} = 
%     \begin{bmatrix}
%     h_i(\bm{x}) \\ 
%     \vdots \\ 
%     L_{f}^{\gamma_i-1}h_i(\bm{x})
%     \end{bmatrix} \in \mathbb{R}^{\gamma_i}.
% \end{align}
the corresponding strict feedback form with respect to the $i$th output $y_i$
% the MIMO system 
can be written as
% \color{red}
% Consider a SISO system of the form $\eqref{eq: system}$ with $m=p=1$, $\gamma$ is
% the relative degree of $y=h(x)$ with $h: \mathbb{R}^{n}\rightarrow \mathbb{R}$.
% We define the transformation as
% \color{black}
% \begin{align}
%     \label{eq: mapping definition}\eta = \begin{bmatrix}\eta_{1}\\ \vdots \\ \eta_{\gamma}\end{bmatrix} = \begin{bmatrix}h(x) \\ \vdots \\ L_{f}^{\gamma-1}h(x)\end{bmatrix} \in \mathbb{R}^{\gamma}
% \end{align}
% the SISO system of form $\eqref{eq: system}$ can be written as
\begin{align}
    \label{eq: strict feedback form for MIMO with uniform relateive degrees}
    \dot{\bm{\eta}}^i % &
    = \begin{bmatrix}
    \dot{\eta}_{1}^i\\ 
    \vdots \\ 
    \dot{\eta}_{\gamma_i}^i
    \end{bmatrix} 
    = \begin{bmatrix}
    \eta_{2}^i\\ 
    \vdots \\ 
    \eta_{\gamma_i}^i\\ 
    \mathcal{L}_{f}^{\gamma_i}h_i(\bm{x}) + \sum_{j=1}^{m}\mathcal{L}_{g_j}(\mathcal{L}_{f}^{\gamma_i -1 }h_i(\bm{x})) u_{j}
    \end{bmatrix} % \\ &=
    % \begin{bmatrix}
    %     0 & 1 & 0 & \cdots & 0 \\
    %     0 & 0 & 1 & \cdots & 0 \\
    %     \vdots & \vdots & \vdots & \ddots & \vdots \\
    %     0 & 0 & 0 & \cdots & 1 \\
    %     0 & 0 & 0 & \cdots & 0
    % \end{bmatrix}
    % % _{r_i \times r_i}
    % \begin{bmatrix}
    %     h(x) \\
    %     \vdots \\
    %     L_f^{\gamma-1} h(x)
    % \end{bmatrix} +
    % \begin{bmatrix}
    %     0 \\
    %     \vdots \\
    %     0 \\
    %     1
    % \end{bmatrix} v
\end{align}
% where $v=L_f^{\gamma} h + \sum_{j=1}^m L_{g_j} (L_f^{\gamma -1 } h) u_j$.
% which is a strict feedback form. 
% Let $\bm{z}_{l}^i= (\eta_{1}^i, \cdots, \eta_{l}^i)$.
% we define the safe output set as
% $\mathcal{C}_{\gamma}= \mathcal{C}\times \prod_{l=1}^{\gamma-1}\mathcal{W}_{l}$
% and $\mathcal{Y}_{\gamma}^{r}= \mathcal{Y}^{r}\times \prod_{l=1}^{\gamma-1}\mathcal{W}
% _{l}$ be the target output set, where $\mathcal{W}_{l} = \{ L_f^l \bm{h}(\bm{x}) \in \mathbb{R}^m | \bm{x} \in \mathcal{C} \}$ represents the set of intermediate variables.
Leveraging this cascaded architecture, we introduce the following ECGBF candidate for MIMO systems \eqref{eq: dynamic system} with the form:

\begin{align}
    \label{eq: RA function for MIMO}
    {\Psi}(\bm{x}) = \psi(\bm{y}(\bm{x})) - \sum_{i=1}^m \sum_{l=1}^{\gamma_i-1}\frac{1}{2\mu_{l}^i}\|\eta_{l+1}^i-k_{l}^i(\bm{z}_{l}^i)\|_{2}^{2}
    % , \forall \bm{x} \in \overline{\mathcal{C} \setminus \mathcal{X}^r}
\end{align}

\noindent where $\bm{z}_{l}^i= (\eta_{1}^i, \cdots, \eta_{l}^i)$, and all $k_l^i(\bm{z}_l^i)$ are auxiliary functions for constructing the actual reach-avoid controller.
% where $\mathcal{C}_\gamma = \{ \bm{x} \in \mathbb{R}^n | \psi_\gamma(\bm{x}) > 0 \}$.
% where $z_{l}= (\eta_{1}, \cdots, \eta_{l})$,
% $\mathcal{C}_{\gamma}= \mathcal{C}\times \prod_{l=1}^{\gamma-1}\mathcal{W}_{l}$
% be the safe output set and $\mathcal{Y}_{\gamma}^{r}= \mathcal{Y}^{r}\times \prod_{l=1}^{\gamma-1}\mathcal{W}
% _{l}$ be the target output set, where $\mathcal{W}_{l} = \{ \}$ is the set of admissible auxiliary variable $\eta_{l+1}$.
% Given this construction, 
Let $\mathcal{C}_\Psi = \{ \bm{x} \in \mathbb{R}^n | \Psi(\bm{x}) >0 \}$.
We summarize our main results
% the condition for ensuring the existence of a reach-avoid controller 
in the following theorem.
\begin{theorem}
    \label{theorem: reach-avoid backstepping for MIMO} 
    If $\psi(\bm{y}(\bm{x}))$ is an ECGBF of system \eqref{eq: single integrator} defined on the safe set $\mathcal{C}$ and target set $\mathcal{X}^{r}$ that satisfies assumption \eqref{assumption on sets},
    then there exists $\lambda>0$ such that the function
    $\Psi(\bm{x}): \mathbb{R}^{n} \rightarrow \mathbb{R}$ defined in \eqref{eq: RA function for MIMO} is an ECGBF for system \eqref{eq: dynamic system} with respect to the safe set $\mathcal{C}_\Psi$ and target set $\mathcal{X}^r$.
\end{theorem}

% where $Q_l$ is symmetric positive definite matrix.

\begin{proof}

    Given $\mathcal{C}$ and $\mathcal{X}^r$ satisfy assumption \eqref{assumption on sets},
    % Next, we verify that the following $\psi_\gamma(\bm{y})$ is an ECGBF for the system \eqref{eq: strict feedback form for MIMO with uniform relateive degrees} on $\overline{\mathcal{C}_\gamma \setminus \mathcal{X}_\gamma^r}$.
    % \begin{align}
    % \label{eq: RA function for MIMO strict feedback form}
    % \psi_{\gamma}(z_\gamma) = \psi(\bm{y}) - \sum_{l=1}^{\gamma-1}\frac{1}{2\mu_{l}}\|\eta_{l+1}-k_{l}(z_{l})\|_{2}^{2}, \quad \forall z_{\gamma}\in \overline{\mathcal{C}_\gamma \setminus \mathcal{Y}_\gamma^r}
    % \end{align}
    % Then we verify that $V_{\gamma}(z_{\gamma})$ is an exponential control guidance-barrier
    % function for the system \eqref{eq: SISO strict feedback form} on
    % $\overline{\mathcal{C}_\gamma \setminus \mathcal{X}_\gamma^r}$.
    by taking the derivative of $\Psi(\bm{x})$ with respect to time, we have
    % \begin{align} \label{eq: derivative of ECGBF}
    %     \dot{{\Psi}}_\gamma(\bm{x}) =& \frac{\partial \psi(y)}{\partial y}\dot{y}\notag \\
    %     &- \sum_{l=1}^{\gamma-1}\frac{1}{\mu_{l}}(\eta_{l+1}- k_{l}(z_{l}))^{\top}(\dot{\eta}_{l+1}-\sum_{s=1}^{l}\frac{\partial k_{l}(z_{l})}{\partial \eta_{s}}\dot{\eta}_{s}) 
    % \end{align}
    
    \small
    \begin{align} \label{eq: derivative of ECGBF}
        \dot{\Psi}(\bm{x}) 
        % =& \frac{\partial \psi(\bm{y})}{\partial \bm{y}}\dot{\bm{y}}\notag \\
        % =& \sum_{i=1}^m \frac{\partial \psi}{\partial y_i} \dot{y}_i\\
        % &- \sum_{i=1}^m \sum_{l=1}^{\gamma_i-1}\frac{1}{\mu_{l}^i}(\eta_{l+1}^i- k_{l}^i)^{\top}(\dot{\eta}_{l+1}^i-\sum_{s=1}^{l}\frac{\partial k_{l}^i}{\partial \eta_{s}^i}\dot{\eta}_{s}^i) \\
        =& \sum_{i=1}^m \bigg[
        % \underbrace{
        \frac{\partial \psi}{\partial y_i} \dot{y}_i
        - \sum_{l=1}^{\gamma_i-1}\frac{(\eta_{l+1}^i- k_{l}^i)^{\top}}{\mu_{l}^i}(\dot{\eta}_{l+1}^i-\sum_{s=1}^{l}\frac{\partial k_{l}^i}{\partial \eta_{s}^i}\dot{\eta}_{s}^i)
        % }_{\Psi_i(\bm{x})}
        \bigg]
    \end{align}
    \normalsize
    For each component $i \in \{1, \cdots, m \}$, by incorporating $k_l^i$ and arranging the corresponding terms\footnote{For conciseness, we abbreviate $k_l^i(z_l^i)$ as $k_l^i$.},
    
    \begin{align*}
        &\frac{\partial \psi}{\partial y_i} \dot{y}_i
        - \sum_{l=1}^{\gamma_i-1}\frac{(\eta_{l+1}^i- k_{l}^i)^{\top}}{\mu_{l}^i}(\dot{\eta}_{l+1}^i-\sum_{s=1}^{l}\frac{\partial k_{l}^i}{\partial \eta_{s}^i}\dot{\eta}_{s}^i) \\
        &= \frac{\partial \psi}{\partial y_i} (\eta_2^i - k_1^i + k_1^i) \\
        &- \sum_{l=1}^{\gamma_i-2}\frac{(\eta_{l+1}^i- k_{l}^i)^{\top}}{\mu_{l}^i}(\dot{\eta}_{l+1}^i - k_l^i + k_l^i -\sum_{s=1}^{l}\frac{\partial k_{l}^i}{\partial \eta_{s}^i}\dot{\eta}_{s}^i) \\
        % &- \frac{(\eta_2^i - k_1^i)^\top}{\mu_1^i}(\eta_3^i - k_2^i + k_2^i - \sum_{s=1}^{1}\frac{\partial k_{1}^i}{\partial \eta_{s}^i}\dot{\eta}_{s}^i)) \\
        % & \cdots \\
        &- \frac{(\eta_{\gamma_i}^i - k_{\gamma_i-1}^i)^\top}{\mu_{\gamma_i-1}^i}
        \bigg[
        \dot{\eta}_{\gamma_i}^i - \sum_{s=1}^{\gamma_i-1}\frac{\partial k_{\gamma_i-1}^i}{\partial \eta_{s}^i}\dot{\eta}_{s}^i
        \bigg] \\
        % \bigg[\mathcal{L}_{f}^{\gamma_i}h_i(\bm{x}) + \sum_{j=1}^{m}\mathcal{L}_{g_j}(\mathcal{L}_{f}^{\gamma_i -1 }h_i(\bm{x})) u_{j} \\ 
        % & \qquad \qquad\qquad\qquad\qquad\qquad\qquad\qquad - \sum_{s=1}^{1}\frac{\partial k_{1}^i}{\partial \eta_{s}^i}\dot{\eta}_{s}^i)\bigg]
    % \end{align*}
    % By arranging the corresponding terms,
    % \begin{align*}
        % & \quad \ \frac{\partial \psi}{\partial y_i} \dot{y}_i
        % - \sum_{l=1}^{\gamma_i-1}\frac{(\eta_{l+1}^i- k_{l}^i)^{\top}}{\mu_{l}^i}(\dot{\eta}_{l+1}^i-\sum_{s=1}^{l}\frac{\partial k_{l}^i}{\partial \eta_{s}^i}\dot{\eta}_{s}^i) \\
        &= \frac{\partial \psi}{\partial y_i} k_1^i - (\eta_2^i - k_1^i)^\top
        \underbrace{
        \bigg[
        - \frac{\partial \psi}{\partial y_i} + \frac{k_2^i - \sum_{s=1}^1 \frac{\partial k_1^i}{\partial \eta_s} \dot{\eta}_s^i }{\mu_1^i}
        \bigg]}_{\mathcal{I}_1} \\
        &- (\eta_3^i - k_2^i)^\top 
        \underbrace{
        \bigg[
        \frac{\eta_2^i - k_1^i}{\mu_1^i} + 
        \frac{k_3^i - \sum_{s=1}^2 \frac{\partial k_2^i}{\partial \eta_s^i} \dot{\eta}_s^i}{\mu_2^i}
        \bigg]}_{\mathcal{I}_2} \\
        & \cdots 
    \end{align*}
\vspace{-0.2cm}
\small
\[
- (\eta_{\gamma_i}^i - k_{\gamma_i-1}^i)^\top
        \underbrace{
        \bigg[
        \frac{\eta_{\gamma_i-1}^i - k_{\gamma_i-2}^i}{\mu_{\gamma_i-2}^i} + 
        \frac{
        % \mathcal{L}_{f}^{\gamma_i}h_i(\bm{x}) + \sum_{j=1}^{m}\mathcal{L}_{g_j}(\mathcal{L}_{f}^{\gamma_i -1 }h_i(\bm{x})) u_{j}
        \dot{\eta}_{\gamma_i}^i
        - \sum_{s=1}^{\gamma_i-1} \frac{\partial k_{\gamma_i-1}}{\eta_s^i} \dot{\eta}_s^i}{\mu_{\gamma_i-1}}
        \bigg]}_{\ \mathcal{I}_{\gamma_i-1}}.
\]
\normalsize
    Given the function $\bm{k}_1(\bm{x})$, we start by picking a $\bm{k}_2(\bm{x})$ such that  $\mathcal{I}_1 = \frac{\lambda (\eta_{2}^i - k_1^i)}{2\mu_1^i}$ by construction. We repeat the process and recursively construct $\mathcal{I}_{l}$ to be equal to $\frac{\lambda (\eta_{l+1}^i - k_l^i)}{2\mu_l^i}$ by appropriately picking $\bm{k}_{l+1}(\bm{x})$ until we get to the last term $\mathcal{I}_{\gamma_i-1}$, where the desired $\bm{k}(\bm{x})$ appears in $\dot{\eta}_{\gamma_i}^i$.
    % $\mathcal{I}_{\gamma_i-1}$
    % , and every $\mathcal{I}_{l+1}$ can be recursively defined by $\mathcal{I}_l$ until $\mathcal{I}_1$, where $\bm{k}_1(\bm{x})$ appears.
    Finally, we define $\bm{k}(\bm{x})$ satisfying these recursive constructions, and its substitution into \eqref{eq: derivative of ECGBF} yields,
    % Let $\bm{k}(\bm{x})$ be defined iteratively such that $\mathcal{I}_l = \frac{\lambda (\eta_{l+1}^i - k_l^i)}{2\mu_l^i}$ for all $l \in \{1, \cdots, \gamma_i-1 \}$. We obtain,
    % Following the definition of the controller $\bm{k}(\bm{x})$ in \eqref{eq: RA controller for MIMO}, its substitution into \eqref{eq: derivative of ECGBF} yields,
    % Given the exponential control guidance-barrier functions as presented in
    % \eqref{eq: RA function for SISO}, assume $k_{1}$ be a smooth function and
    % define $k_{i}$ for $i \in \{3,\cdots, \gamma-1\}$, $k_{u}$ as same as $\cref{eq:
    % k_2,eq: k_3+,eq: k_u}$, we have
    % \begin{align*}
    %     \dot{\psi}_{\gamma}(\bm{x}) &= \frac{\partial \psi}{\partial y}\bm{k}_{1}(\bm{y}(\bm{x})) - \sum_{l=1}^{\gamma-1}\frac{\lambda}{2\mu_{l}}\| \eta_{l+1}- k_{l}(z_{l})\|_{2}^{2} \\
    %     &\geq \lambda \psi(\bm{y}(\bm{x})) - \sum_{l=1}^{\gamma-1}\frac{\lambda}{2\mu_{l}}\| \eta_{l+1}- k_{l}(z_{l})\|_{2}^{2} \\
    %     &= \lambda \psi_{\gamma}(\bm{x})
    % \end{align*}
    \begin{align*}
        \dot{\Psi}(\bm{x}) &= \frac{\partial \psi}{\partial \bm{y}}\bm{k}_{1}(\bm{y}(\bm{x})) 
        - \sum_{i=1}^m \sum_{l=1}^{\gamma_i-1}\frac{\lambda}{2\mu_{l}^i}\| \eta_{l+1}^i- k_{l}^i(\bm{z}_{l})\|_{2}^{2} \\
        &\geq \lambda \psi(\bm{y}(\bm{x})) - \sum_{i=1}^m \sum_{l=1}^{\gamma_i-1}\frac{\lambda}{2\mu_{l}^i}\| \eta_{l+1}^i- k_{l}^i(\bm{z}_{l}^i)\|_{2}^{2} \\
        &= \lambda \Psi(\bm{x}),
    \end{align*}
    which implies that 
    \begin{align}
        \sup_{\bm{u}}\mathcal{L}_{\Psi,\bm{u}} \geq \mathcal{L}_{\Psi,\bm{k}(\bm{x})} \geq \lambda \Psi(\bm{y}(\bm{x})), \ \forall \bm{x} \in \overline{\mathcal{C} \setminus \mathcal{X}^r}.
    \end{align}
    The positivity of $\Psi(\bm{y})$ ensures $\psi(\bm{y})>0$ by implication, thereby the zero superlevel set of $\Psi(\bm{y})$ essentially defines a subset of the given safe set $\mathcal{C}$. Thus,
    \begin{align} \label{eq: subset ECGBF constraint}
        \sup_{\bm{u}}\mathcal{L}_{\Psi,\bm{u}} \geq \lambda \Psi(\bm{y}(\bm{x})), \ \forall \bm{x} \in \overline{\mathcal{C}_\Psi \setminus \mathcal{X}^r}
    \end{align}
    where $\mathcal{C}_\Psi = \{ \bm{x} \in \mathbb{R}^n | \Psi(\bm{x}) >0 \}$.
    % Thus, for all $\bm{y} \in \$
    % As $V_{\gamma}$ is an exponential control guidance-barrier function on $\mathcal{C}
    % _{\gamma}$, this implies that $\dot{V}_{\gamma}\geq \lambda V_{\gamma}$, thus
    % $\forall y \in \overline{\mathcal{C}_\gamma \setminus \mathcal{X}_\gamma^r}$,
    % we have
    % \begin{align*}
    %     \frac{\partial \psi}{\partial y}k_{1}(y) + \sum_{l=1}^{\gamma-1}\frac{\lambda}{2\mu_{l}}\| \eta_{l+1} & - k_{l}(z_{l})\|_{2}^{2}                                                             \\
    %                                                                                                           & \geq                                                                                 \\
    %     \lambda \big[\psi(y)                                                                                  & + \sum_{l=1}^{\gamma-1}\frac{1}{2\mu_{l}}\| \eta_{l+1}- k_{l}(z_{l}) \|\big]_{2}^{2} % \\
    %     % \Rightarrow \frac{\partial \psi}{\partial y} k_1(y) \geq \lambda \psi(y)
    % \end{align*}
    % This means that there exists a controller $k_{1}$ and $\lambda >0$ such that
    % \begin{align}
    %     \frac{\partial \psi}{\partial y}k_{1}(y) \geq \lambda \psi(y), \quad \forall y \in \overline{\mathcal{C}_\gamma \setminus \mathcal{X}_\gamma^r}\label{eq: k_1 existence}
    % \end{align}
    Additionally, since $\mathcal{C}_\Psi$ can arbitrarily approximate $\mathcal{C}$ by adjusting each $\mu_l^i$, 
    and $\Psi(x)$ is constructed to be continuous such that $\mathcal{C}_\Psi$ is an open set without isolated points,
    % if $\mathcal{C}_\Psi$ has isolated point, then there exists a point $\bm{x}_q \in \mathcal{C}_\Psi \subseteq \mathcal{C}$ and $\xi>0$ such that
    % \begin{align}
    %     \mathcal{C}_{\Psi}\cap \{ \mathcal{N}\big(\bm{x}_q, \xi \big) \setminus \{ \bm{x}_q \} \} = \emptyset
    % \end{align}
    % which implies $\bm{x}_q$ is an isolated point of $\mathcal{C}$ satisfying
    % \begin{align}
    %     \mathcal{C} \cap \{ \mathcal{N}\big(\bm{x}_q, \xi \big) \setminus \{ \bm{x}_q \} \} = \emptyset
    % \end{align}
    % and leads to a contradiction. Thus 
    thus $\mathcal{C}_\Psi$ satisfy the Assumption \eqref{assumption on sets}.
    By Theorem \eqref{theorem: reach-avoid controller synthesis} we can conclude that $\Psi(\bm{x})$ is also an ECGBF for the system \eqref{eq: dynamic system} with respect to the safe set $\mathcal{C}_\Psi$ and the target set $\mathcal{X}^r$.
    Furthermore, the synthesized controller $\bm{k}(\bm{x})$ is a reach-avoid controller with provable guarantees that all trajectories $\bm{x}(t)$ originating from $\mathcal{C}_\Psi$ will eventually enter the target set $\mathcal{X}^r$ while maintaining $\bm{x}(t) \in \mathcal{C}_\Psi$.
\end{proof}

We note that the proof of theorem \ref{theorem: reach-avoid backstepping for MIMO} is a constructive one, and it establishes not just a valid ECGBF candidate $\Psi(\bm{x})$ but also provides a way to obtain the reach-avoid controller recursively from $\bm{k_1}(\bm{x})$ using backstepping procedure. In practice, one can obtain this controller $\bm{k}(\bm{x})$ without actually constructing the ECGBF $\Psi(\bm{x})$. We present this result in the following proposition. 

\begin{proposition}
    Given that $\Psi(\bm{x}): \mathbb{R}^n \rightarrow \mathbb{R}$ as defined in \eqref{eq: RA function for MIMO} is an ECGBF for system \eqref{eq: dynamic system} with respect to $\mathcal{C}_\Psi$ and target set $\mathcal{X}^r$, any locally Lipschitz controller $\bm{k}(\bm{x}) \in \mathcal{K}(\bm{x})$ is a reach-avoid controller where
    \begin{align}
        \mathcal{K} (\bm{x}) = \{ \bm{k}(\bm{x}) \in \mathbb{R}^m \ | \ \text{constraints} \ \eqref{eq: RA controller for MIMO} \ \text{holds} \}
    \end{align} with
    
    \begin{equation}
        \label{eq: RA controller for MIMO}
        \left\{
        \begin{split}
            & \frac{\partial \psi(\bm{y}(\bm{x}))}{\partial \bm{y}} \cdot \bm{k}_1(\bm{y}(\bm{x})) \geq \lambda \psi(\bm{y}(\bm{x})), \forall \bm{x} \in \overline{\mathcal{C} \setminus \mathcal{X}^r} \\
            % & \nabla_{\bm{y}} \psi(\bm{y}) \cdot \bm{k}_1(\bm{y}) \geq \lambda \psi(\bm{y}), \forall \bm{y} \in \overline{\mathcal{C} \setminus \mathcal{Y}^r} \\
            % \bm{k}_2(\bm{y}) =& - \mu_{1}(\frac{\partial \psi(y)}{\partial y})^{\top} + \sum_{s=1}^{1}\frac{\partial k_{1}}{\partial \eta_{s}}\eta_{s+1} + \frac{\lambda}{2}(\eta_{2}- k_{1}) \\
            % \bm{k}_2(\bm{z}_2) &= \mu_{1}\frac{\partial \psi(\bm{y}(\bm{x}))}{\partial \bm{y}} + \sum_{s=1}^{1}\frac{\partial k_{1}}{\partial \eta_{s}}\eta_{s+1} + \frac{\lambda}{2}(\eta_{2}- k_{1}) \\
            \bm{k}_2^i(\bm{z}^i_2) &= \mu_{1}^i\frac{\partial \psi(\bm{y}(\bm{x}))}{\partial \bm{y}_i} + \sum_{s=1}^{1}\frac{\partial k_{1}^i}{\partial \eta_{s}^i}\eta_{s+1}^i + \frac{\lambda}{2}(\eta_{2}^i- k_{1}^i) \\
            % \bm{k}_{i}(\bm{y}) =& - \frac{\mu_{i-1}(\eta_{i-1}- k_{i-2})}{\mu_{i-2}} + \sum_{s=1}^{i-1}\frac{\partial k_{i-1}}{\partial \eta_{s}}\eta_{s+1} \\
            % & \qquad \qquad + \frac{\lambda}{2}(\eta_{i} - k_{i-1}), \quad \forall i \in \{1, \cdots, \gamma-1\} \\
            % \bm{k}_{i}(\bm{z}_i) &= - \frac{\mu_{i-1}(\eta_{i-1} - k_{i-2})}{\mu_{i-1}} + \sum_{s=1}^{i-1} \frac{\partial k_{i-1}}{\partial \eta_s} \eta_{s+1} \\
            % & \qquad \quad + \frac{\lambda}{2} (\eta_i - k_{i-1}), \quad \forall i \in \{3, \cdots, \gamma-1\}\\
            \bm{k}_l^i(\bm{z}_l^i) &= - \frac{\mu_{l-1}^i(\eta_{l-1}^i - k_{l-2}^i)}{\mu_{l-1}^i} + \sum_{s=1}^{i-1} \frac{\partial k_{l-1}^i}{\partial \eta_s} \eta_{s+1}^i \\
            & \qquad \quad + \frac{\lambda}{2} (\eta_l^i - k_{l-1}^i), \quad \forall l \in \{3, \cdots, \gamma_i-1\}\\
            % \bm{k}(\bm{y})=& \textcolor{red}{G}^{\dagger}\bigg[ - \textcolor{red}{F} - \frac{\mu_{\gamma-1}(\eta_{\gamma-1}- k_{\gamma-2})}{\mu_{\gamma-2}} \\
            % & \qquad \qquad \quad + \sum_{s=1}^{\gamma-1}\frac{\partial k_{\gamma-1}}{\partial \eta_{s}}\eta_{s+1}+ \frac{\lambda}{2}(\eta_{\gamma}- k_{\gamma-1}) \bigg] \\
            % \bm{k}(\bm{x}) &= \bm{G}(\bm{x})^{\dagger}\bigg[ - L_f^\gamma h(x) - \frac{\mu_{\gamma-1}(\eta_{\gamma-1}- k_{\gamma-2})}{\mu_{\gamma-2}} \\
            % & \qquad \quad \quad + \sum_{s=1}^{\gamma-1}\frac{\partial k_{\gamma-1}}{\partial \eta_{s}}\eta_{s+1}+ \frac{\lambda}{2}(\eta_{\gamma}- k_{\gamma-1}) \bigg] \\
            % \bm{A}(\bm{x}) &=
            % \begin{bmatrix}
            %     L_{g_1}(L_f^{\gamma-1}h(x)) & \cdots & L_{g_m}(L_f^{\gamma-1}h(x))
            % \end{bmatrix}
            % \bm{A}(\bm{x}) &\text{ is defined as in \eqref{matrix A(x)}} \\
            % _{m \times m}
            % \bm{b}_i(\bm{x}) &= - L_f^\gamma h(x) - \frac{\mu_{\gamma-1}(\eta_{\gamma-1}- k_{\gamma-2})}{\mu_{\gamma-2}} \\
            % & \qquad \quad \quad + \sum_{s=1}^{\gamma-1}\frac{\partial k_{\gamma-1}}{\partial \eta_{s}}\eta_{s+1}+ \frac{\lambda}{2}(\eta_{\gamma}- k_{\gamma-1}) \\
            b_i(\bm{x}) &= - \mathcal{L}_f^{\gamma_i} h_i(\bm{x}) - \frac{\mu_{\gamma_i-1}^i(\eta_{\gamma_i-1}^i- k_{\gamma_i-2}^i)}{\mu_{\gamma-2}} \\
            & \qquad \quad \quad + \sum_{s=1}^{\gamma_i-1}\frac{\partial k_{\gamma_i-1}^i}{\partial \eta_{s}^i}\eta_{s+1}^i+ \frac{\lambda}{2}(\eta_{\gamma_i}^i- k_{\gamma_i-1}^i) \\
            \lambda > & 0, \mu_l^i >0, \forall l \in \{1, \cdots, \gamma_i-1\}, \forall i \in \{1, \cdots, m\}. \\
            \bm{b}(\bm{x}) &= \begin{bmatrix}
                b_1(\bm{x}) & \cdots & b_m(\bm{x})
            \end{bmatrix}^\top \\
            \bm{A}(\bm{x}) &\text{ is defined as in \eqref{matrix A(x)}} \\
            \bm{k}(\bm{x}) &= \bm{A}^{-1}(\bm{x}) \bm{b}(\bm{x})
        \end{split}
        \right.
    \end{equation}
\end{proposition}
\begin{proof}
    By the construction of each $\mathcal{I}_l$ in Theorem \ref{theorem: reach-avoid backstepping for MIMO}'s proof, $\bm{k}(\bm{x})$ can be recursively defined as in constraints \eqref{eq: RA controller for MIMO}.
\end{proof}

\begin{remark}
    Theorem \eqref{theorem: reach-avoid backstepping for MIMO} demonstrates that we can achieve reach-avoid guarantees for the safe subset $\mathcal{C}_\Psi \subseteq \mathcal{C}$ and target set $\mathcal{X}^r$ by constructing a controller $\bm{k}(\bm{x})$ based on a single-integrator system \eqref{eq: single integrator} and leveraging the strict feedback form \eqref{eq: strict feedback form for MIMO with uniform relateive degrees}.
    The proposed reach-avoid controller synthesis approach avoids solving SOS programs for the original system \eqref{eq: dynamic system}, instead requiring only optimization on \eqref{eq: single integrator}, which is computationally lighter and scalable for high-dimensional systems.
    Additionally, 
    % corollary \eqref{corollary: on equivalence} reveals that \eqref{eq: subset ECGBF constraint} implies the satisfaction of \eqref{eq: ECGBF constraint} as $\mathcal{C}_\Psi$ can be an inner approximation of $\mathcal{C}$ with arbitrary precision by increasing $\mu_l^i$.
    % However, 
    excessively large $\mu_l^i$ may induce rapid changes of input control signals, which could lead to instability in both simulations and real-world deployments.
    Conversely, overly small $\mu_l^i$ can shrink the zero superlevel set of $\Psi(\bm{x})$, potentially limiting the applicability of the synthesized controller and may result in an empty interaction with the target set $\mathcal{X}^r$ thereby violating assumption \eqref{assumption on sets}.
    % derived from Theorem \eqref{theorem: reach-avoid backstepping for uniform MIMO} inherently satisfies constraint \eqref{eq: on equivalence}
    % It is noticed that the positivity of $\psi_\gamma(\bm{y})$ ensures $\psi(\bm{y})>0$ by implication, thereby establishing that the zero superlevel set of $\psi_\gamma(\bm{y})$ defines a subset of the given safe set $\mathcal{C}$. In the absence of practical constraints on input control signals, $\mu_l$ can be adjusted to achieve arbitrary approximation of $\mathcal{C}$ by the zero superlevel set of $\psi_\gamma(\bm{y})$.
\end{remark}

\begin{example}
    Consider the following control-affine system
    \small
    \begin{align} \label{eq: ex1 system}
    % \dot{\bm{x}} &=
        % \begin{bmatrix}
        %     \dot{x}_1 \\
        %     \dot{x}_2 \\
        %     \dot{x}_3 \\
        %     \dot{x}_4
        % \end{bmatrix} 
        \dot{\bm{x}}
        =  
        \begin{bmatrix}
            x_2 \\
            -x_2^2+x_3^3+x_4-5x_1 \\
            x_4 \\
            2x_2x_3+2x_3x_4-7x_1
        \end{bmatrix} &+
        \begin{bmatrix}
            0  \\
            2x_2^2+5x_3x_4-9x_1 \\
            0 \\
            2x_2x_3^3+2x_3x_4+7x_2
        \end{bmatrix} u_1 \notag \\
        &+ 
        \begin{bmatrix}
            0 \\
            3x_2^3+2x_3x_4 -10x_1 \\
            0 \\
            7x_2^2x_3+2x_3x_4
        \end{bmatrix} u_2
        % \bm{f}(\bm{x}) + \bm{g}_1(\bm{x}) u_1 + \bm{g}_2(\bm{x}) u_2
        % \\
    % \dot{\bm{x}} &=
    %     \begin{bmatrix}
    %         \dot{x}_1 \\
    %         \dot{x}_2 \\
    %         \dot{x}_3 \\
    %         \dot{x}_4
    %     \end{bmatrix} = 
    %     \begin{bmatrix}
    %         x_2 \\
    %         f_0(\bm{x}) \\
    %         x_4 \\
    %         f_1(\bm{x})
    %     \end{bmatrix} + 
    %     \begin{bmatrix}
    %         0 & 0 \\
    %         g_0^0(x) & g_0^1(x) \\
    %         0 & 0 \\
    %         g_1^0(x) & g_1^1(x) 
    %     \end{bmatrix} 
    %     % \begin{bmatrix}
    %     %     0 & 0 \\
    %     %     2x_2^2+5x_3x_4-9x_1 & g_0^1(x) \\
    %     %     0 & 0 \\
    %     %     g_1^0(x) & g_1^1(x) 
    %     % \end{bmatrix} 
    %     \begin{bmatrix}
    %         u_0 \\
    %         u_1
    %     \end{bmatrix}
    %     \\
        % \bm{y} &= 
        % \begin{bmatrix}
        %     x_1 \\
        %     x_3
        % \end{bmatrix}
    \end{align}
    \normalsize
    with $\bm{x} = [x_1, x_2, x_3, x_4]^\top \in \mathbb{R}^4$, $u_1,u_2 \in \mathbb{R}$ and output
    $\bm{y} = [x_1, x_3]^\top \in \mathbb{R}^2$.
    % $y_1(\bm{x}), y_2(\bm{x}) \in \mathbb{R}$. 

    \begin{figure}[htbp!]
    \centering
    \includegraphics[width=0.9\linewidth]{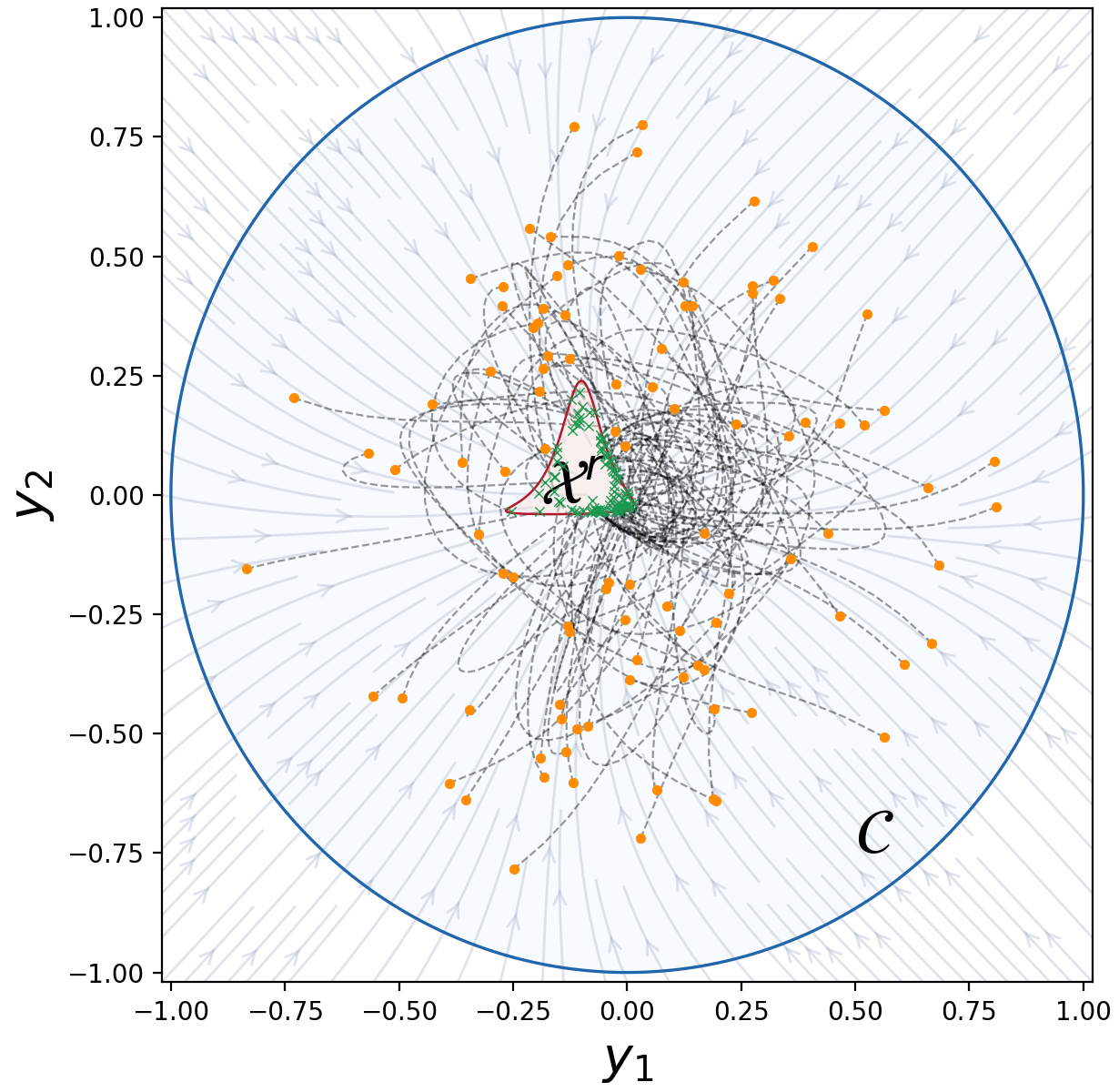}
    \caption{$100$ randomly sampled trajectories (dashed line from orange to green cross) of system \eqref{eq: ex1 system} with synthesized reach-avoid controller. The grey vector-field indicates the function $\bm{k}_1(\bm{y}(\bm{x}))$, from which we construct $\bm{k}(\bm{y}(\bm{x}))$ via backstepping.}
    \label{fig:ex00}
    \end{figure}
    
    The system has a vector relative degree $\{2,2\}$. Given the safe set $\mathcal{C} = \{ \bm{x} \in \mathbb{R}^4 | 1- \bm{y}_1(\bm{x})^2 - \bm{y}_2(\bm{x})^2 < 0 \}$ and target set $\mathcal{X}^r = \{ \bm{x} \in \mathbb{R}^4 | 2 (\frac{\bm{y}_2(\bm{x})-0.1}{2})^2 + 3(\frac{\bm{y}_1(\bm{x})+0.4}{3})^4 + (3\bm{y}_1(\bm{x})+0.3)^2 (4\bm{y}_2(\bm{x})+0.2)^2 < 0.01 \}$, the control objective is to design a state feedback controller that ensures the output trajectories remain within the safe set $\mathcal{C}$ and eventually enter into the target set $\mathcal{X}^r$. 
    % Let $\eta = 
    % \begin{bmatrix}
    %     h(x) &
    %     L_f h(x)
    % \end{bmatrix}^\top$, the system can be transformed to the form
    % \begin{align}
    %     \dot{\eta} = 
    %     \begin{bmatrix}
    %         L_f h(x) \\
    %         L_f^2 h(x) + L_gL_f h u
    %     \end{bmatrix}
    % \end{align}
    By solving the optimization problem \eqref{algo: sos reach-avoid controller}, we obtain the controller $\bm{k}_1(\bm{y}(\bm{x}))$ with $\delta=-3.23 \times 10^{-10} \approx 0$, demonstrating that $\bm{k}_1(\bm{y}(\bm{x}))$ is a reach-avoid controller that guarantees the compliance of state $\bm{y}$ of the single-integrator system \eqref{eq: single integrator} with the required reach-avoid specifications. 
    By applying the backstepping-based methodology outlined in \eqref{theorem: reach-avoid backstepping for MIMO}, we synthesize the reach-avoid controller $\bm{k}(\bm{x})$ using $\bm{k}_1(\bm{y}(\bm{x}))$, and implement it in system \eqref{eq: ex1 system} to form the corresponding closed-loop system.
    We randomly sampled $100$ initial points within the set $\mathcal{C}_\Psi$ to validate the performance of our proposed approach. As depicted in Fig. \ref{fig:ex00}, all resulting output trajectories remain within the safe set $\mathcal{C}$ during evolution and ultimately reach the target set $\mathcal{X}^r$. 
    Furthermore, Fig. \ref{fig:ex01} demonstrates that the function $\Psi(\bm{x})$ increases monotonically along all trajectories, confirming that $\Psi(\bm{x})$ serves as an ECGBF for the system \eqref{eq: ex1 system} with respect to safe set $\mathcal{C}$ and target set $\mathcal{X}^r$.

    \begin{figure}[htbp!]
    \centering
    \includegraphics[width=0.9\linewidth]{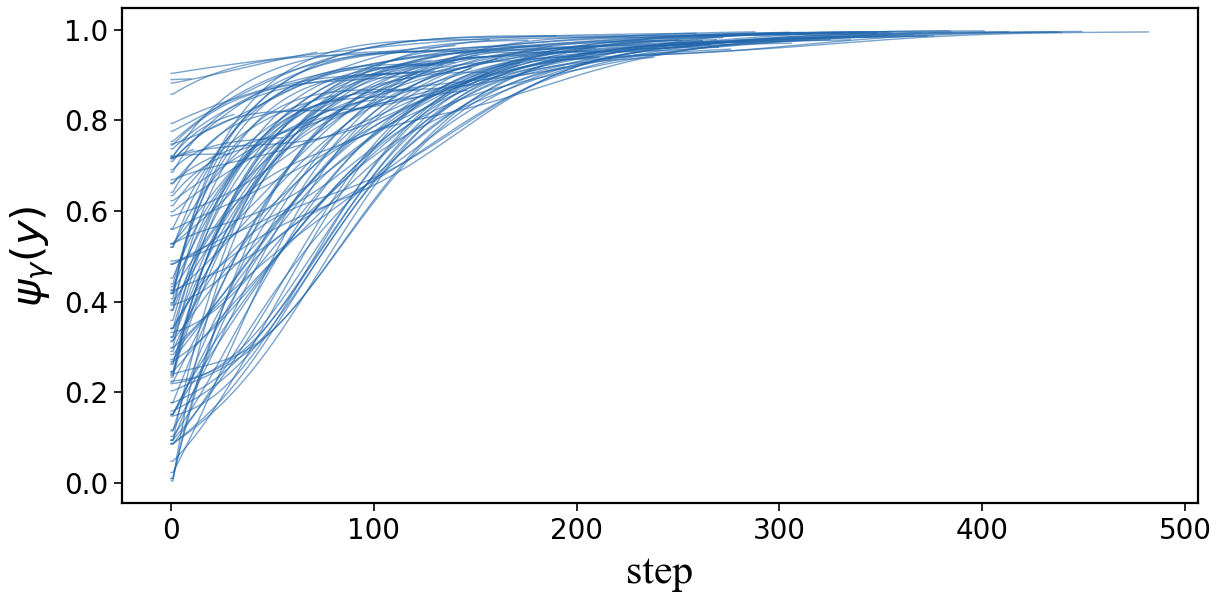}
    \caption{Evaluation of $\Psi(\bm{x})$ along $100$ randomly sampled trajectories.}
    \label{fig:ex01}
    \end{figure}

\end{example}

%% file: results.tex
\section{Experiments}

In this section, we numerically validate the theoretical developments of the proposed methodology on two underactuated dynamical systems. 
The SOSTOOLS \cite{sostools} toolbox with Mosek \cite{aps2019mosek} solver is employed to formulate and solve SOS optimization problems.
% The associated SOS optimization problems are formulated and solved by the SOSTOOLS \cite{sostools} toolbox with Mosek \cite{aps2019mosek} solver.

% Building on the previous example, we further demonstrate our approach using two case studies: control of a Dubins car and a 2-DoF planar link.

% \subsection{Dubins Car Track Field Example}
% \subsection{Dubins Car}

\begin{example} 
% (Dubins Car)
\label{ex: dubins car}
\begin{figure}
    \centering
    \includegraphics[width=0.9\linewidth]{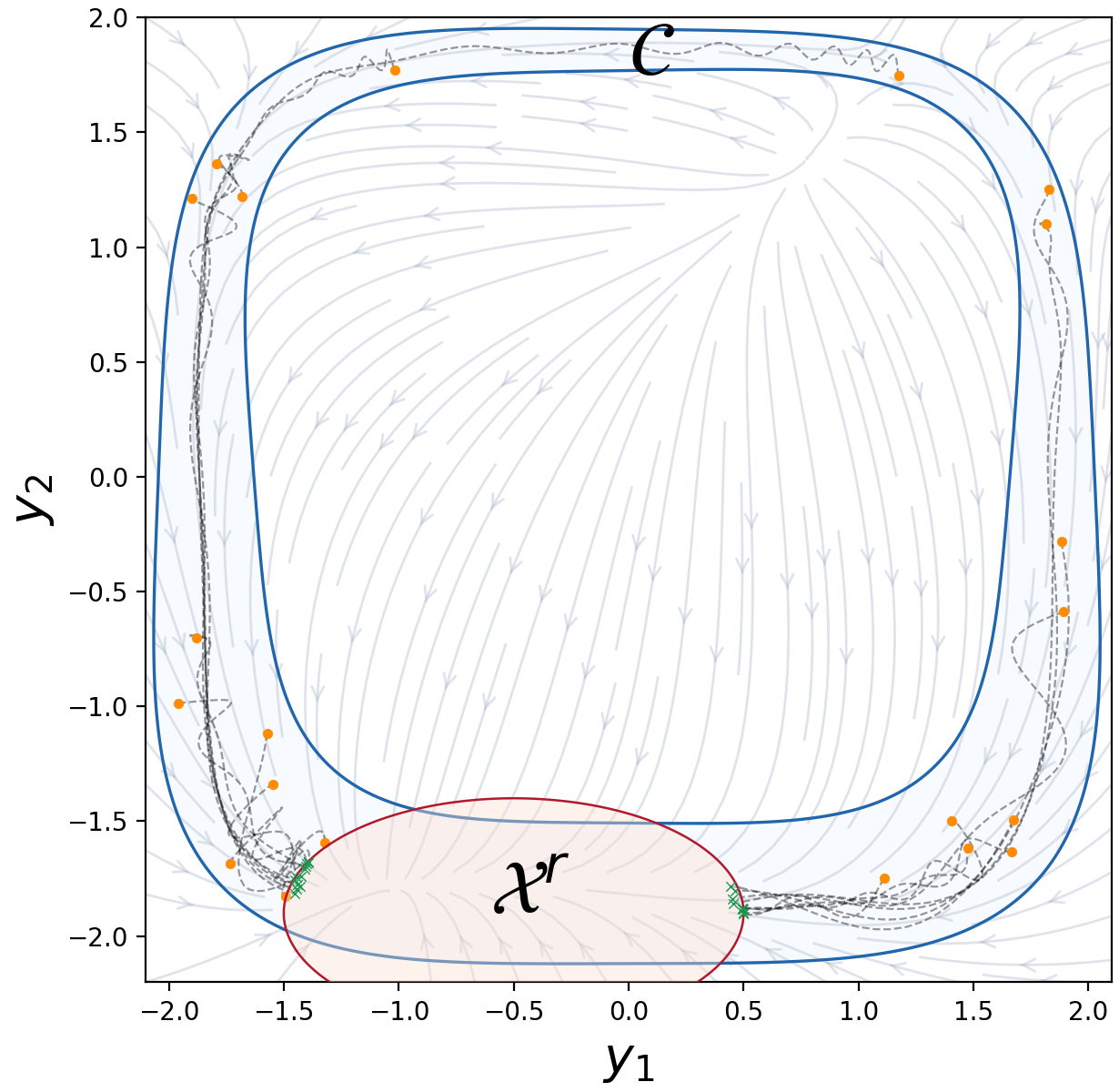}
    \caption{Simulated trajectories (dashed line from orange dot to green cross) of Dubins car on a racetrack with synthesized reach-avoid controller. The grey vector field indicates the function $\bm{k}_1(\bm{y}(\bm{x}))$, from which we construct $\bm{k}(\bm{y}(\bm{x}))$ via backstepping.}
    \label{fig:dubins}
\end{figure}

Consider the following modified Dubins car model
\begin{equation}
    \dot{\bm{x}} = 
    \left[\begin{array}{c}\dot{x}_1\\\dot{x}_2\\\dot{\theta}\\\dot{v}\end{array}\right]= 
    \underbrace{\left[\begin{array}{c}v\cos{\theta}\\v\sin{\theta}\\0\\0\end{array}\right]}_{f(\bm{x})} + \underbrace{\left[\begin{array}{cc}0 & 0\\0 & 0\\1 & 0\\0&1\end{array}\right]}_{g(\bm{x})}
    \underbrace{\left[\begin{array}{c}\omega\\a\end{array}\right]}_{\bf{u}},
    \label{eq.dubins_car}
\end{equation}
with $\bm{x} = [x_1,x_2,\theta,v]^\top \in \mathbb{R}^4$, where $(x_1,x_2)$ denotes the planar location, $\theta$ represents the heading angle, and $v$ is the forward velocity. The vehicle is actuated by the angular velocity $\omega$ and forward acceleration $a$ as control inputs $\bm{u} = [\omega, a]^\top \in \mathbb{R}^2$. The system has a vector relative degree $\{2,2\}$ with output $\bm{y} = [y_1, y_2]^\top = [x_1, x_2]^\top$.
Given the safe set $\mathcal{C} = \{\bm{x} \in \mathbb{R}^2 | 2(3.5 - y_2(\bm{x}))^2 - 1.5y_1(\bm{x}) - (y_1(\bm{x})^4 + y_2(\bm{x})^4 - 3.5^2)^2 > 0 \}$ and target set $\mathcal{X}^r= \{ \bm{x} \in \mathbb{R}^2 | (y_1(\bm{x})+0.5)^2+4(y_2(\bm{x})+1.9)^2 < 1\}$.
For simplicity, we neglect the size of the vehicle. Our goal is to generate a feedback controller that drives the vehicle into the target set $\mathcal{X}^r$ while preventing trajectory violations of the track boundaries throughout the entire maneuver.
As shown in Fig. \ref{fig:dubins}, the SOS optimized controller $\bm{k}_1(\bm{y}(\bm{x}))$ ensures the vector field of single-integrator system \eqref{eq: single integrator} converges to the target set $\mathcal{X}^r$ while strictly respecting the safety constraints.
The simulated trajectories from randomly sampled initial states demonstrate that the synthesized controller $\bm{k}(\bm{x})$ drives the dubins vehicle \eqref{eq.dubins_car} into the target set without violating the prescribed track boundaries before reaching.

\end{example}

\begin{example} 
\label{ex: robotic arm}
% (Two-link Planar Robotic Arm)
% Consider an Acrobot, which is essentially a planar two-link robotic arm, with the following dynamics,
Consider the following two-link planar robotic arm,
% in the $xz-$plane,
% with following dynamics,
\begin{equation}
    \dot{\bm{q}} =
    % \left[\begin{array}{c}
    % \dot{q}_1\\
    % \dot{q}_2\\
    % \ddot{q}_1\\
    % \ddot{q}_2
    % \end{array}\right] = 
    \underbrace{\left[\begin{array}{c}
    \dot{q}_1\\
    \dot{q}_2\\
    -\bm{M}^{-1}(\bm{C}\dot{\bm{q}}+\bm{G})
    \end{array}\right]}_{f(\bm{q})}+\underbrace{\left[\begin{array}{c}\bm{0}^{2\times2}\\\bm{M}^{-1}\end{array}\right]}_{g(\bm{q})}\underbrace{\left[\begin{array}{c}
    \tau_1\\
    \tau_2\end{array}\right]}_{\bm{u}}.
\end{equation}
with $\bm{q} = [q_1,q_2,\dot{q}_1,\dot{q}_2]^\top \in \mathbb{R}^4$, and $[q_1, q_2]^\top$ as the joint angles, and $[\dot{q}_1, \dot{q}_2]^\top$ denoting the joint velocities. The inertia matrix $\bm{M}$ is defined as
\[\bm{M} = \begin{bmatrix}
    M_{11} & M_{12}\\
    M_{21} & M_{22}
\end{bmatrix},\]
where $M_{11} = I_1 + I_2 + m_2(l_1^2 + l_{c_2}^2 + 2l_1l_{c_2}\cos{(q_2)})+m_1l_{c_1}^2$, $M_{12} = M_{21} =  I_2 + m_2(l_{c_2}^2 + l_1l_{c_2}\cos{(q_2)})$, $M_{22} = m_2l_{c_2}^2 + I_2$.
\[
\bm{C} =
\begin{bmatrix}
- m_2 l_1 l_{c_2} \sin(q_2) \dot{q}_2 & - m_2 l_1 l_{c_2} \sin(q_2) (\dot{q}_1 + \dot{q}_2) \\
m_2 l_1 l_{c_2} \sin(q_2) \dot{q}_1 & 0
\end{bmatrix}
\] is the Coriolis matrix. The centrifugal matrix,
\[
\bm{G} =
\begin{bmatrix}
(m_1 g l_{c_1} + m_2 g l_1) \cos(q_1) + m_2 g l_{c_2} \cos(q_1 + q_2) \\
m_2 g l_{c_2} \cos(q_1 + q_2)
\end{bmatrix}\]
denotes the gravity vector. $\bm{\tau}\in\mathbb{R}^2$ is the joint torque as control inputs.
Other parameters are listed in Table \ref{tab: parameters}.
\setlength\extrarowheight{1pt}
\begin{table}[htbp] 
    \centering
    \aboverulesep = 0pt
    \belowrulesep = 0pt
    \caption{Parameters of the 2-link planar robotic arm}
    \begin{adjustbox}{width=0.48\textwidth,center,keepaspectratio}
    \begin{tabular}{|M{0.03\textwidth}|M{0.03\textwidth}|M{0.03\textwidth}|M{0.03\textwidth}|M{0.03\textwidth}|M{0.03\textwidth}|M{0.03\textwidth}|M{0.03\textwidth}|M{0.05\textwidth}|}
        \toprule
         \multicolumn{2}{|c|}{Link Mass}  & \multicolumn{2}{c|}{Link Length} & \multicolumn{2}{c|}{Mass center} & \multicolumn{2}{c|}{Inertia moment} & Gravity \\
         \midrule
         $m_1$ & $m_2$ & $l_1$ & $l_2$ & $l_{c_1}$ & $l_{c_2}$ & $I_1$ & $I_2$ & $g$ \\
         \midrule
         1.0 & 1.0 & 4.0 & 4.0 & 2.0 & 2.0 & 0.02 & 0.02 & 9.81\\
         \bottomrule
    \end{tabular}
    \end{adjustbox}
    \label{tab: parameters}
\end{table}

% \begin{table}[htbp]
%     \centering
%     \aboverulesep = 0pt
%     \belowrulesep = 0pt
%     %\begin{adjustbox}{width=0.48\textwidth,center,keepaspectratio}
%     \begin{tabular}%{|c|c|c|c|c|}
% {|M{0.1\textwidth}|M{0.064\textwidth}|M{0.05\textwidth}|}
%     \toprule
%          % \multicolumn{2}{|c|}{\multirow{2}{*}{Example}} & \multicolumn{2}{c|}{Time (s)} \\
%          % \multicolumn{3}{|c|}{Example} & \multicolumn{2}{c|}{Compute time (s)} \\
%          % \cmidrule{3-4}
%          % \midrule 
%          Parameter & Variable & Value  \\
%          \midrule
%         \multirow{2}{*}{Link mass} & $m_1$ & 1.0 \\
%          \cmidrule{2-3}
%          & $m_2$ & 1.0 \\
%          \midrule
%          \multirow{2}{*}{Link length} & $l_1$ & 4.0 \\
%          \cmidrule{2-3}
%          & $l_2$ & 4.0 \\
%          \midrule
%          \multirow{2}{*}{Mass center} & $l_{c_1}$ & 2.0 \\
%          \cmidrule{2-3}
%          & $l_{c_2}$ & 2.0 \\
%          \midrule
%          \multirow{2}{*}{Inertia moment} & $I_1$ & 0.02 \\
%          \cmidrule{2-3}
%          & $I_2$ & 0.02 \\
%          \midrule
%          Gravity & g & 9.81 \\
%          \bottomrule
%     \end{tabular}
%     %\end{adjustbox}
%     \caption{}
%     \label{tab: parameters ex2}
% \end{table}

\begin{figure}[htbp]
    \centering
    \includegraphics[width=0.9\linewidth]{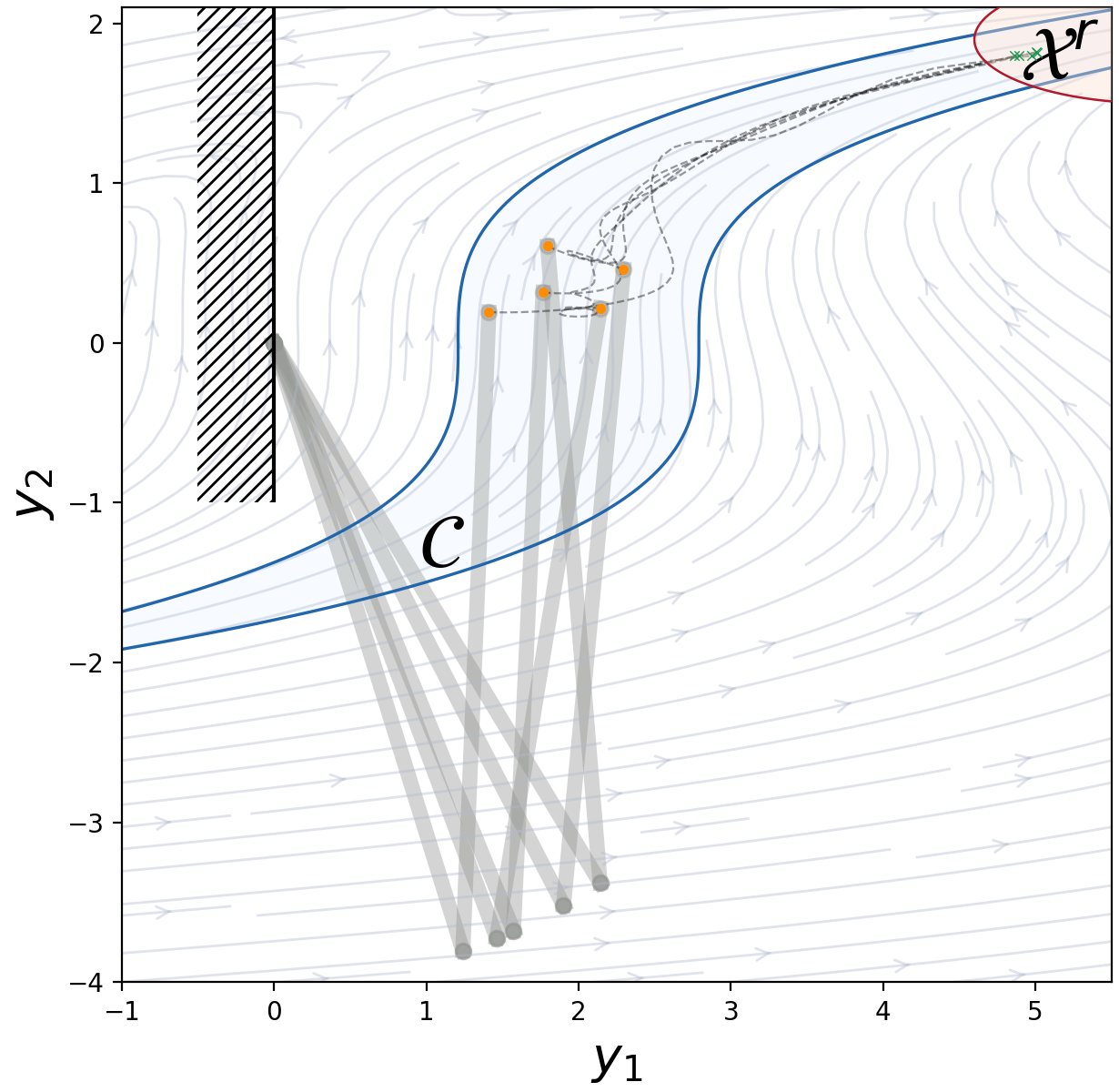}
    \caption{Simulated trajectories (dashed line from orange dot to green cross) of 2-link planar robotic arm. The grey vector field indicates the function $\bm{k}_1(\bm{y}(\bm{x}))$, from which we construct $\bm{k}(\bm{y}(\bm{x}))$ via backstepping.}
    \label{fig:planar_link03}
\end{figure}

The system has a vector relative degree $\{2,2 \}$ with the output $\bm{y} \in \mathbb{R}^2$ be the endpoint position defined as
\begin{equation}
    \dot{\bm{y}} = \left[\begin{array}{c}y_1\\y_2\end{array}\right] = \left[\begin{array}{c}
    l_1\cos{(q_1)} + l_2\cos{(q_1 + q_2)}\\
    l_1\sin{(q_1)} + l_2\sin{(q_1 + q_2)}
    \end{array}\right].
\end{equation}

Given safe set $\mathcal{C}= \{ \bm{q} \in \mathbb{R}^4 | -(4(y_1(\bm{q})-2)-2y_2(\bm{x})^3)^2 + 0.8y_2(\bm{q})^3 + 10 >0 \}$ and target set $\mathcal{X}^r = \{ \bm{q} \in \mathbb{R}^4 | ((y_1(\bm{q})- 5.8)^2 / 1.2^2) + ((y_2(\bm{q}) - 1.9)^2 / 0.4^2) - 1 <0 \}$,
we aim to design a controller that enables the robotic arm to execute desired motion within specified error bounds, ultimately achieving predefined terminal configurations.
Unlike traditional nominal trajectory tracking tasks, we introduce the safe set $\mathcal{C}$ governing admissible deviations during motion execution, and the target set $\mathcal{X}^r$ defining allowable final motion tolerances.
As evidenced by the simulated output trajectories in Fig. \ref{fig:planar_link03}, the synthesized controller successfully drives the arm from randomly sampled feasible initial configurations within admissible error bounds, ultimately converging to the target motions, while rigorously maintaining motion feasibility throughout the entire operation.

\end{example}

% \begin{remark}
To illustrate the impact of different choices of $\mu_l^i$ in \eqref{eq: RA function for MIMO}, 
Fig. \ref{fig:twofigures} visualizes the zero super level sets of $\Psi(\bm{x})$ in examples \ref{ex: dubins car} and \ref{ex: robotic arm} by setting $[\theta,v]$ and $[\dot{q}_1, \dot{q}_2]$ as $\bm{0}$, respectively.
% we visualize the initial feasible starting regions for the two examples as the green regions within their respective original safe sets $\mathcal{C}$ in Fig. \ref{fig:twofigures}. The shading of the green region represents the magnitude of $\mu_l^i$, with lighter shades indicating larger $\mu_l^i$ values. 
% As observed, increasing $\mu_l^i$ results in a larger initial feasible set $\mathcal{C}_\Psi$, which provides a closer approximation to the original safe set $\mathcal{C}$.  
As observed, increasing $\mu_l^i$ results in a larger $\mathcal{C}_\Psi \subseteq \mathcal{C}$, wherein all states maintain provable reach-avoid guarantees with synthesized controller $
\bm{k}(\bm{x})$.

\begin{figure}[ht]
    \centering
    \begin{subfigure}{0.49\linewidth}
        \centering
        \includegraphics[width=\textwidth]{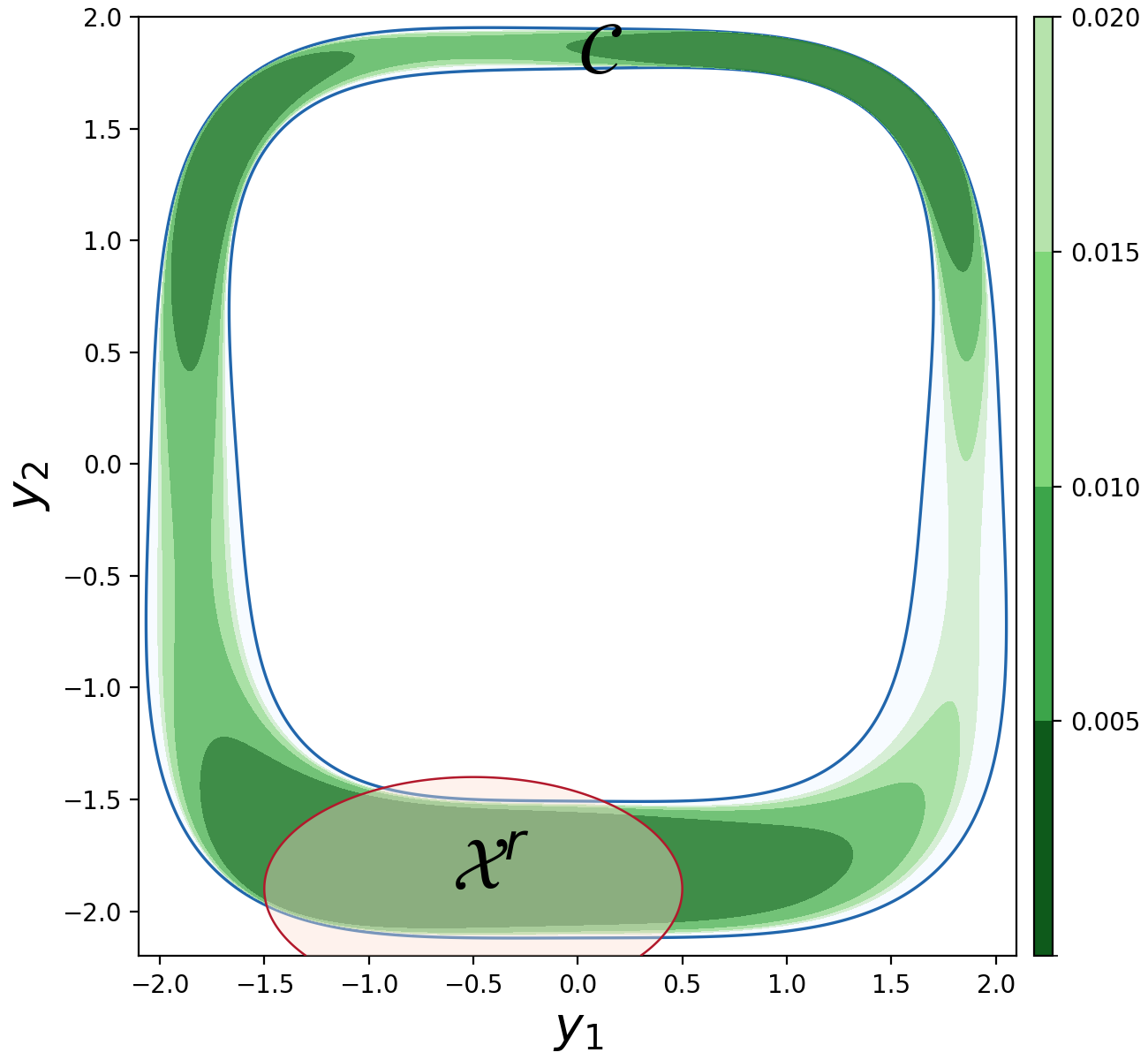}
        \caption{Evolving of $\mathcal{C}_\Psi$ in example \ref{ex: dubins car}}
        \label{fig:sub1}
    \end{subfigure}
    \hfill
    \begin{subfigure}{0.49\linewidth}
        \centering
        \includegraphics[width=\textwidth]{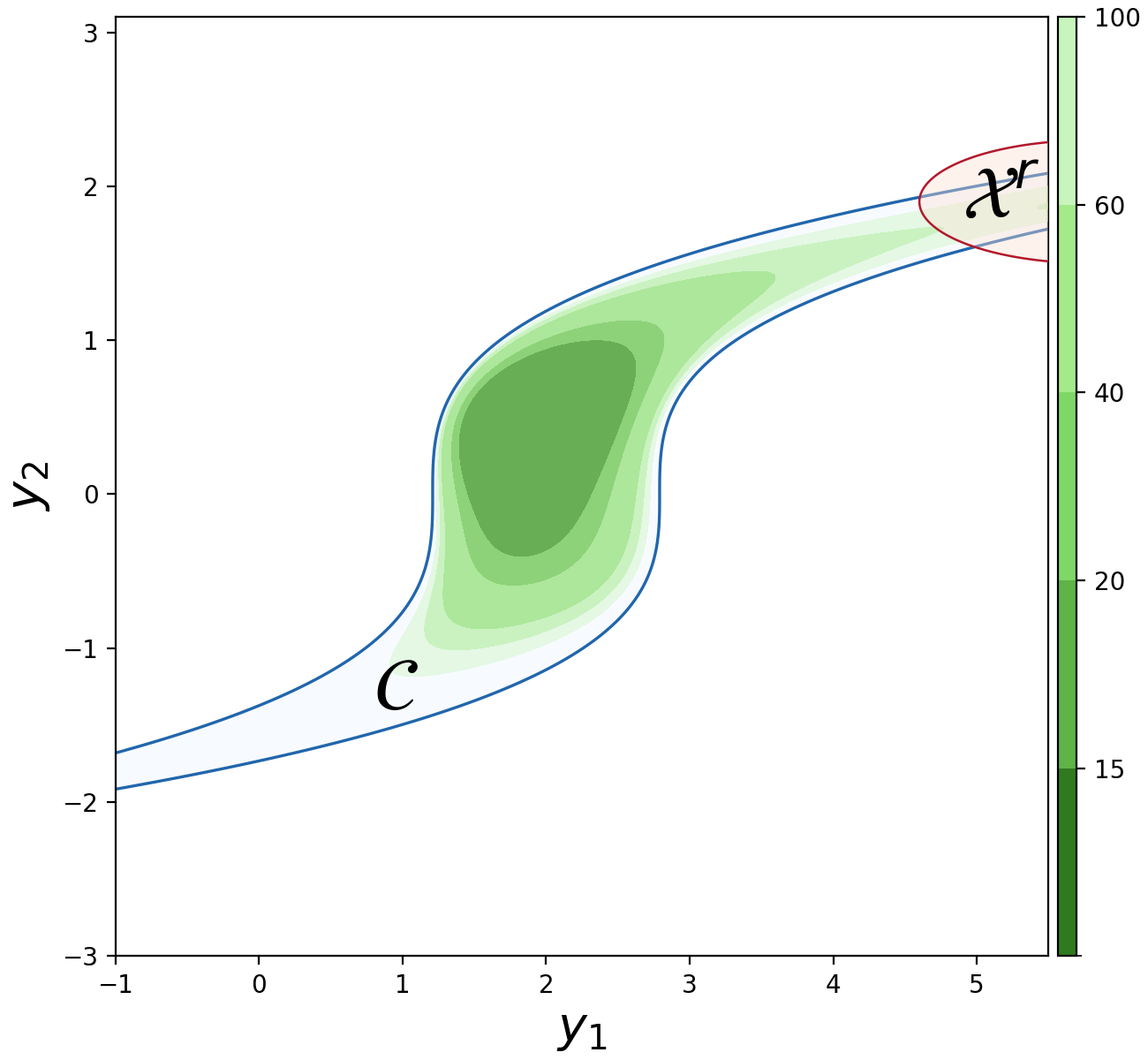}
        \caption{Evolving of $\mathcal{C}_\Psi$ in example \ref{ex: robotic arm}}
        \label{fig:sub2}
    \end{subfigure}
    \caption{
    % Initial starting set $\mathcal{C}_\Psi$ for different choices of $\mu_l^i$ values for the above two examples.
    Impact of increasing $\mu_l^i$ on $\mathcal{C}_\Psi$. Color bar indicates the different values of $\mu_l^i$.
    % the zero superlevel sets of $\Psi(\bm{x})$.
    }
    \label{fig:twofigures}
\end{figure}

%% file: conclusion.tex
\section{Conclusion}
In this paper, we present a novel framework for constructing formal certificates and corresponding controllers with provable reach-avoid guarantees.
% The proposed certificates characterize a subset of the given safe set, where any state within it can achieve the specified reach-avoid objectives with synthesized controllers.
The proposed method simplifies the design of reach-avoid controllers for nonlinear MIMO systems with mixed relative degrees by first finding a fundamental reach-avoid controller via SOS optimization, then systematically synthesizing the desired actual controller through feedback linearization and backstepping techniques, thereby establishing an efficient design framework.
We finally validate the effectiveness of our method on numerical simulations.
Future work will focus on enhancing robustness in practical scenarios by incorporating considerations for noise, and extending the framework's applicability to more general non-polynomial 
% dynamic systems.
settings.

%% file: main.bbl
\begin{thebibliography}{10}

\bibitem{choi2021robust}
J.~J. Choi, D.~Lee, K.~Sreenath, C.~J. Tomlin, and S.~L. Herbert, ``Robust control barrier--value functions for safety-critical control,'' in {\em 2021 60th IEEE Conference on Decision and Control (CDC)}, pp.~6814--6821, IEEE, 2021.

\bibitem{ames2016control}
A.~D. Ames, X.~Xu, J.~W. Grizzle, and P.~Tabuada, ``Control barrier function based quadratic programs for safety critical systems,'' {\em IEEE Transactions on Automatic Control}, vol.~62, no.~8, pp.~3861--3876, 2016.

\bibitem{ames2019control}
A.~D. Ames, S.~Coogan, M.~Egerstedt, G.~Notomista, K.~Sreenath, and P.~Tabuada, ``Control barrier functions: Theory and applications,'' in {\em 2019 18th European control conference (ECC)}, pp.~3420--3431, Ieee, 2019.

\bibitem{nguyen2021robust}
Q.~Nguyen and K.~Sreenath, ``Robust safety-critical control for dynamic robotics,'' {\em IEEE Transactions on Automatic Control}, vol.~67, no.~3, pp.~1073--1088, 2021.

\bibitem{deka2018robust}
S.~A. Deka, X.~Li, D.~M. Stipanovi{\'c}, and T.~Kesavadas, ``Robust and safe coordination of multiple robotic manipulators: An approach using modified avoidance functions,'' {\em Journal of Intelligent \& Robotic Systems}, vol.~90, pp.~419--435, 2018.

\bibitem{wabersich2023data}
K.~P. Wabersich, A.~J. Taylor, J.~J. Choi, K.~Sreenath, C.~J. Tomlin, A.~D. Ames, and M.~N. Zeilinger, ``Data-driven safety filters: Hamilton-jacobi reachability, control barrier functions, and predictive methods for uncertain systems,'' {\em IEEE Control Systems Magazine}, vol.~43, no.~5, pp.~137--177, 2023.

\bibitem{zeng2021safety}
J.~Zeng, B.~Zhang, and K.~Sreenath, ``Safety-critical model predictive control with discrete-time control barrier function,'' in {\em 2021 American Control Conference (ACC)}, pp.~3882--3889, IEEE, 2021.

\bibitem{yuan2024safe}
D.~Yuan, X.~Yu, S.~Li, and X.~Yin, ``Safe-by-construction autonomous vehicle overtaking using control barrier functions and model predictive control,'' {\em International Journal of Systems Science}, vol.~55, no.~7, pp.~1283--1303, 2024.

\bibitem{zhang2023efficient}
H.~Zhang, Z.~Li, H.~Dai, and A.~Clark, ``Efficient sum of squares-based verification and construction of control barrier functions by sampling on algebraic varieties,'' in {\em 2023 62nd IEEE Conference on Decision and Control (CDC)}, pp.~5384--5391, IEEE, 2023.

\bibitem{bansal2017hamilton}
S.~Bansal, M.~Chen, S.~Herbert, and C.~J. Tomlin, ``Hamilton-jacobi reachability: A brief overview and recent advances,'' in {\em 2017 IEEE 56th Annual Conference on Decision and Control (CDC)}, pp.~2242--2253, IEEE, 2017.

\bibitem{chen2025distributionally}
Y.~Chen, Y.~Li, S.~Li, and X.~Yin, ``Distributionally robust control synthesis for stochastic systems with safety and reach-avoid specifications,'' {\em arXiv preprint arXiv:2501.03137}, 2025.

\bibitem{vaidyanathan2020backstepping}
S.~Vaidyanathan and A.~T. Azar, {\em Backstepping control of nonlinear dynamical systems}.
\newblock Academic Press, 2020.

\bibitem{taylor2022safe}
A.~J. Taylor, P.~Ong, T.~G. Molnar, and A.~D. Ames, ``Safe backstepping with control barrier functions,'' in {\em 2022 IEEE 61st Conference on Decision and Control (CDC)}, pp.~5775--5782, IEEE, 2022.

\bibitem{kim2023safe}
J.~Kim and Y.~Kim, ``Safe control synthesis for multicopter via control barrier function backstepping,'' in {\em 2023 62nd IEEE Conference on Decision and Control (CDC)}, pp.~8720--8725, IEEE, 2023.

\bibitem{koga2023safe}
S.~Koga and M.~Krstic, ``Safe pde backstepping qp control with high relative degree cbfs: Stefan model with actuator dynamics,'' {\em IEEE Transactions on Automatic Control}, vol.~68, no.~12, pp.~7195--7208, 2023.

\bibitem{nguyen2016exponential}
Q.~Nguyen and K.~Sreenath, ``Exponential control barrier functions for enforcing high relative-degree safety-critical constraints,'' in {\em 2016 American Control Conference (ACC)}, pp.~322--328, IEEE, 2016.

\bibitem{xiao2019control}
W.~Xiao and C.~Belta, ``Control barrier functions for systems with high relative degree,'' in {\em 2019 IEEE 58th conference on decision and control (CDC)}, pp.~474--479, IEEE, 2019.

\bibitem{xiao2021high}
W.~Xiao and C.~Belta, ``High-order control barrier functions,'' {\em IEEE Transactions on Automatic Control}, vol.~67, no.~7, pp.~3655--3662, 2021.

\bibitem{cohen2024constructive}
M.~H. Cohen, R.~K. Cosner, and A.~D. Ames, ``Constructive safety-critical control: Synthesizing control barrier functions for partially feedback linearizable systems,'' {\em IEEE Control Systems Letters}, 2024.

\bibitem{xu2018constrained}
X.~Xu, ``Constrained control of input--output linearizable systems using control sharing barrier functions,'' {\em Automatica}, vol.~87, pp.~195--201, 2018.

\bibitem{abel2023prescribed}
I.~Abel, D.~Steeves, M.~Krsti{\'c}, and M.~Jankovi{\'c}, ``Prescribed-time safety design for strict-feedback nonlinear systems,'' {\em IEEE Transactions on Automatic Control}, vol.~69, no.~3, pp.~1464--1479, 2023.

\bibitem{dawson2022learning}
C.~Dawson, B.~Lowenkamp, D.~Goff, and C.~Fan, ``Learning safe, generalizable perception-based hybrid control with certificates,'' {\em IEEE Robotics and Automation Letters}, vol.~7, no.~2, pp.~1904--1911, 2022.

\bibitem{hsu2021safety}
K.-C. Hsu, V.~Rubies-Royo, C.~J. Tomlin, and J.~F. Fisac, ``Safety and liveness guarantees through reach-avoid reinforcement learning,'' {\em arXiv preprint arXiv:2112.12288}, 2021.

\bibitem{xue2024reach}
B.~Xue, ``Reach-avoid controllers synthesis for safety critical systems,'' {\em IEEE Transactions on Automatic Control}, 2024.

\bibitem{xue2023reach}
B.~Xue, N.~Zhan, M.~Fr{\"a}nzle, J.~Wang, and W.~Liu, ``Reach-avoid verification based on convex optimization,'' {\em IEEE Transactions on Automatic Control}, vol.~69, no.~1, pp.~598--605, 2023.

\bibitem{isidori1985nonlinear}
A.~Isidori, {\em Nonlinear control systems: an introduction}.
\newblock Springer, 1985.

\bibitem{sostools}
A.~Papachristodoulou, J.~Anderson, G.~Valmorbida, S.~Prajna, P.~Seiler, and P.~A. Parrilo, {\em {SOSTOOLS}: Sum of squares optimization toolbox for {MATLAB}}.
\newblock \texttt{http://arxiv.org/abs/1310.4716}, 2013.

\bibitem{aps2019mosek}
M.~ApS, ``Mosek optimization toolbox for matlab,'' {\em User’s Guide and Reference Manual, Version}, vol.~4, no.~1, p.~116, 2019.

\end{thebibliography}
